\documentclass[12pt,draftcls, peerreview, onecolumn]{IEEEtran}

\usepackage[usenames]{color}
\usepackage{url}
\usepackage{cite}      

\usepackage{graphicx}  

\usepackage{subfigure} 

\usepackage{url}       

\usepackage{amsmath}   

\newtheorem{theorem}{Theorem}[section]

\renewcommand{\baselinestretch}{2}

\begin{document}

\title{
{An MGF-based Unified Framework to Determine the Joint Statistics of Partial Sums of Ordered Random Variables}}

\author{Sung~Sik~Nam,~\IEEEmembership{Member,~IEEE,}
	Mohamed-Slim~Alouini,~\IEEEmembership{Fellow~Member,~IEEE,}
	and Hong-Chuan~Yang,~\IEEEmembership{Senior~Member,~IEEE.}
\thanks{$^{\star}$
This work was supported in part by Qatar Telecom (Qtel). This is an extended version of a paper which was presented in Proc. of IEEE International Conference on Wireless Communications and Signal Processing (WCSP 2009), Nanjing, China, November 2009.
S.~S.~Nam was with Department of Electrical and Computer Engineering, Texas A\&M University at College Station, Texas, USA. He is now with Department of Electronic Engineering in Hanyang University, Seoul, Korea. M.-S.~Alouini was with the Electrical and Computer Engineering Program at Texas A\&M University at Qatar, Doha, Qatar. He is now with Electrical Engineering Program, KAUST, Thuwal, Saudi Arabia. H.~-C.~Yang is with Department of Electrical and Computer Engineering, University of Victoria, BC V8W 3P6, Canada.  They can be reached by E-mail at
\texttt{$<$ssnam11@tamu.edu, slim.alouini@kaust.edu.sa, hyang@ece.uvic.ca$>$}.}}



\markboth{S.S. Nam \MakeLowercase{\textit{et al.}}: An Unified Framework ... on Ordered Statistics}{Shell \MakeLowercase{\textit{et al.}}: An Unified Framework ... on Ordered Statistics}

\maketitle

\begin{abstract}
Order statistics find applications in various areas of communications and signal
processing. In this paper, we introduce an unified analytical framework to determine
the joint statistics of partial sums of ordered random variables (RVs). With the proposed approach, we can systematically derive the joint statistics of any partial sums of ordered statistics, in terms of the moment generating function (MGF) and the probability density function (PDF). Our MGF-based approach applies not only when all the $K$ ordered RVs are involved but also when only the $K_s$ $(K_s <K)$ best RVs are considered. In addition, we present the closed-form expressions for the exponential RV special case. These results apply to the performance analysis of various wireless communication systems over fading channels. 
\end{abstract}



\section{Introduction}
The subject of order statistics deals with the properties and distributions of the ordered random variables (RVs) and their functions. It has found applications in  many areas of statistical theory and practice~\cite{kn:Order_Statistics}, with examples
including life-testing, quality control, signal and image processing~\cite{kn:Handbook_of_Statistics,kn:Goldsmith_book}. Recently, order statistics makes a growing number of appearance in the analysis and design of wireless communication systems (see for example~\cite{kn:simon_book, kn:MS_GSC, kn:OT_MRC_j, kn:Choi_finger, kn:Hong_Sumrate, kn:alouini_wi_j3, kn:ma00, kn:win2001, kn:annamalai2002, kn:Mallik_05, kn:bouida2008}). 
For example, diversity techniques have been used for over the past fifty years to mitigate the effects of fading on wireless communication systems. These techniques  improve the performance of wireless systems over fading channels by generating and combining multiple replicas of the same information bearing signal at the receiver. The analysis of low-complexity selection combining schemes, which select the best replica, requires some basic results of order statistics, i.e. the distribution functions of the largest one among several random variables.

More recently, the design and analysis of adaptive diversity combining techniques and multiuser scheduling strategies call for some further results on order statistics \cite{kn:MS_GSC, kn:OT_MRC_j}. In particular, the joint statistics of partial sums of ordered RVs are often necessary for the accurate characterization of system performance~\cite{kn:Choi_finger,kn:bouida2008}. The major difficulty in obtaining the statistics of partial sums of ordered RVs resides in the fact that even if the original unordered RVs are independently distributed, their ordered versions are necessarily dependent due to the inequality relations among them. Recently, the co-author has applied a successive conditioning approach to convert dependent ordered random variables to independent unordered ones \cite{kn:MS_GSC, kn:OT_MRC_j}. That approach, however, requires some case-specific manipulations, which may not always be generalizable. 

In this paper, we present an unified analytical framework to determine the joint statistics of partial sums of ordered RVs using a moment generating functions (MGF) based approach. More specifically, we extend the result in~\cite{kn:Integral_Solution_Nuttall,kn:joint_PDF_Nuttall,kn:Multi_PDF_Order_Selection}, which only derive the joint MGF of the selected individual order statistics and the sum of the remaining ones, and systematically solve for the joint statistics of arbitrary partial sums of ordered RVs. The main advantage of the proposed MGF-based unified framework is that it applies not only to the cases when all the $K$ ordered RVs are considered but also to those cases when only the $K_s$ $(K_s <K)$ best RVs are involved. After considering several illustrative examples, we focus on the exponential RV special case and derive the closed-form expressions of the joint statistics. These statistical results can apply to the performance analysis of various wireless communication systems over generalized fading channels.

{The remainder of this paper is organized as follows. In section II, we summarize the main idea behind the proposed unified analytical framework, including the general idea and some special considerations. We then introduce some common functions and useful relations in section III, which will help make the results in later sections more compact. In section IV and V, we present some selected examples on the derivation of joint PDF based on our proposed approach. Following this, we show in section VI some closed form expressions for the selected examples presented in previous sections under i.i.d. Rayleigh fading conditions. Finally, we discuss some useful applications of these results in section VII.}

\section{The Main Idea}
Let $\infty \ge \gamma_{1:K} \ge \gamma_{2:K} \ge \gamma_{3:K} \cdots \ge \gamma_{K:K} \ge 0$ be the order statistics obtained by arranging $K$ nonnegative i.i.d. RVs, $\left\{ {\gamma _i } \right\}_{i = 1}^K$, in decreasing order of magnitude. The objective is to derive the joint
PDF of their partial sums involving either all $K$ or the first $K_s$($K_s < K$) ordered RVs, e.g. the
joint PDF of $\gamma _{m:K}$ and $\sum\limits_{\scriptstyle n = 1 \hfill \atop \scriptstyle n \ne m \hfill}^K {\gamma _{n:K} } $ or the joint PDF of $\sum\limits_{n = 1}^m {\gamma _{n:K} }$ and $\sum\limits_{n = m + 1}^{K_s } {\gamma _{n:K} }$.

\subsection{General steps}
The proposed analytical framework adopts a general two-step approach: i) obtain the analytical expressions of the joint MGF of partial sums (not necessarily the partial sums of interest as will be seen later), ii) apply inverse Laplace transform to derive the joint PDF of partial sums (additional integration may be required to obtain the desired joint PDF). To facilitate the inverse Laplace transform calculation, the joint MGF from step i) should be made as compact as possible. An observation made in~\cite{kn:Integral_Solution_Nuttall,kn:joint_PDF_Nuttall,kn:Multi_PDF_Order_Selection} involving the interchange of multiple integrals of ordered RVs becomes useful in the following analysis.
Suppose for example that we need to evaluate a multiple integral over the range $\gamma_a\ge\gamma_1\ge\gamma_2\ge\gamma_3\ge\gamma_4\ge\gamma_b$. More specifically, let
\begin{equation} \label{eq:multi_integral}
\small I \! = \! \int_{\gamma _b }^{\gamma _a } {d\gamma _1 \int_{\gamma _b }^{\gamma _1 } {d\gamma _2 \int_{\gamma _b }^{\gamma _2 } {d\gamma _3 \int_{\gamma _b }^{\gamma _3 } {d\gamma _4 \;\; p\left( {\gamma _1 ,\gamma _2 ,\gamma _3 ,\gamma _4 } \right)} } } }.
\end{equation}
It can be shown by interchanging the order of integration, while ensuring each pair of limits is chosen to be as tight as possible, the multiple integral in (\ref{eq:multi_integral}) can be rewritten into the following equivalent representations, \small
\begin{eqnarray} \label{eq:Integral_solution}
I \!\!\!\! &=& \!\!\!\! \int_{\gamma _b }^{\gamma _a } {d\gamma _4 \int_{\gamma _4 }^{\gamma _a } {d\gamma _3 \int_{\gamma _3 }^{\gamma _a } {d\gamma _2 \int_{\gamma _2 }^{\gamma _a } {d\gamma _1 \;\; p\left( {\gamma _1 ,\gamma _2 ,\gamma _3 ,\gamma _4 } \right)} } } } \nonumber
\\
 \!\!\!\! &=& \!\!\!\! \int_{\gamma _b }^{\gamma _a } {d\gamma _2 \int_{\gamma _b }^{\gamma _2 } {d\gamma _3 \int_{\gamma _b }^{\gamma _3 } {d\gamma _4 \int_{\gamma _2 }^{\gamma _a } {d\gamma _1 \;\; p\left( {\gamma _1 ,\gamma _2 ,\gamma _3 ,\gamma _4 } \right)} } } }  \nonumber
\\
 \!\!\!\! &=& \!\!\!\! \int_{\gamma _b }^{\gamma _a } {d\gamma _3 \int_{\gamma _3 }^{\gamma _a } {d\gamma _1 \int_{\gamma _b }^{\gamma _3 } {d\gamma _4 \int_{\gamma _3 }^{\gamma _1 } {d\gamma _2 \;\; p\left( {\gamma _1 ,\gamma _2 ,\gamma _3 ,\gamma _4 } \right)} } } }  \nonumber
\\
 \!\!\!\! &=& \!\!\!\! \int_{\gamma _b }^{\gamma _a } {d\gamma _1 \int_{\gamma _b }^{\gamma _1 } {d\gamma _4 \int_{\gamma _4 }^{\gamma _1 } {d\gamma _3 \int_{\gamma _3 }^{\gamma _1 } {d\gamma _2 \;\; p\left( {\gamma _1 ,\gamma _2 ,\gamma _3 ,\gamma _4 } \right)} } } }  \nonumber
\\
 \!\!\!\! &=& \!\!\!\! \int_{\gamma _b }^{\gamma _a } {d\gamma _4 \int_{\gamma _4 }^{\gamma _a } {d\gamma _1 \int_{\gamma _4 }^{\gamma _1 } {d\gamma _3 \int_{\gamma _3 }^{\gamma _1 } {d\gamma _2 \;\; p\left( {\gamma _1 ,\gamma _2 ,\gamma _3 ,\gamma _4 } \right)} } } }.
\end{eqnarray} \normalsize
The general rule is that the integration limits should be selected as tight as possible using the remaining variables. For example, in the first equation of (\ref{eq:Integral_solution}), the variables are integrated in the order of $\gamma_1$, $\gamma_2$, $\gamma_3$, and $\gamma_4$. Based on the given inequality condition $\gamma_a\ge\gamma_1\ge\gamma_2\ge\gamma_3\ge\gamma_4\ge\gamma_b$, the integration limit of $\gamma_1$ should be from $\gamma_2$ to $\gamma_a$ because $\gamma_2$ is the tightest among the remaining RVs. Similarly,  the integration limit $\gamma _3$ is from $\gamma_4$ to $\gamma_a$ because $\gamma_1$ and $\gamma_2$ were already integrated out.

After obtaining the joint MGF in a compact form, we can derive joint PDF of selected partial sum through inverse Laplace transform. For most cases of our interest, the joint MGF involves basic functions, for which the inverse Laplace transform can be calculated analytically. In the worst case, we may rely on the Bromwich contour integral. In most of the case, the result involves a single one-dimensional contour integration, which can be easily and accurately evaluated numerically with the help of integral tables~\cite{kn:Mathematical_handbook, kn:gradshteyn_6} or using standard mathematical packages such as Mathematica and Matlab.

\subsection{Special cases}

The general steps can be directly applied when all $K$ ordered RVs are considered and the RVs in the partial sums are continuous. When these conditions do not hold, we need to apply some extra steps in the analysis in order to obtain a valid joint MGF. Specifically, 
when only the best $K_s$ ($K_s < K$) ordered RVs are involved in the partial sums, we should consider the $K_s$th order statistics $\gamma_{K_s:K}$ separately. Without this separation, we cannot find the valid integration limit in calculating the joint MGF. As an illustration, the example in~Fig.~\ref{Example_1} considers 3-dimensional joint PDF of the partial sums of three group of RVs which are $\{\gamma_{1:K},\gamma_{2:K},\gamma_{3:K}\}$, $\{\gamma_{4:K},\gamma_{5:K},\gamma_{6:K}\}$, and $\{\gamma_{7:K},\gamma_{8:K}\}$, with $K=10$ and $K_s=8$. Following the proposed approach, we will derive the 4-dimensional joint MGF in step i) by considering $\gamma_{8:K}$ separately, i.e. the joint MGF of the following four groups of RVs $\{\gamma_{1:K},\gamma_{2:K},\gamma_{3:K}\}$, $\{\gamma_{4:K},\gamma_{5:K},\gamma_{6:K}\}$, $\{\gamma_{7:K}\}$, and $\{\gamma_{8:K}\}$. After obtaining the corresponding 4-dimensional joint MGF, we need to perform another finite integration to obtain the desired 3-dimensional joint PDF.

When the RVs involved in one partial sum is not continuous, i.e. separated by the other RVs, we need to divide these RVs into smaller sums. The example in ~Fig.~\ref{Example_2} illustrate this process. If we consider 3-dimensional joint PDF of $\{\gamma_{1:K}$, $\gamma_{2:K}$, $\gamma_{5:K}$, $\gamma_{6:K}\}$, $\{\gamma_{3:K},\gamma_{4:K}\}$, and $\{\gamma_{7:K},\gamma_{8:K}\}$, then there are three partial sum of RVs $\{\gamma_{1:K}$, $\gamma_{2:K}$, $\gamma_{5:K}$, $\gamma_{6:K}\}$, $\{\gamma_{3:K},\gamma_{4:K}\}$, and $\{\gamma_{7:K},\gamma_{8:K}\}$. Note that the first group is split by the second group of RVs (as such discontinuous). More specifically, the second group, $\{\gamma_{3:K},\gamma_{4:K}\}$, split the original group, $\{\gamma_{1:K}$, $\gamma_{2:K}$, $\gamma_{5:K}$, $\gamma_{6:K}\}$, into two groups as $\{\gamma_{1:K}, \gamma_{2:K}\}$, and $\{\gamma_{5:K}, \gamma_{6:K}\}$. Therefore, we also consider this split group $\{\gamma_{1:K}, \gamma_{2:K}\}$ and $\{\gamma_{5:K}, \gamma_{6:K}\}$ as two smaller groups. As a result, we will derive 5-dimensional joint MGF of  $\{\gamma_{1:K},\gamma_{2:K}\}$, $\{\gamma_{3:K},\gamma_{4:K}\}$, $\{\gamma_{5:K},\gamma_{6:K}\}$, $\{\gamma_{7:K}\}$, and $\{\gamma_{8:K}\}$ in step i). Similarly to the first example, after the joint PDF of the new substituted partial sums are derived with inverse Laplace transform in step ii), we can transform it to a lower dimensional desired joint PDF with finite integration. 

The proposed approach is summarized in the flowchart given in Fig.~\ref{flowchart}, where we consider three different case separately. In the following sections, we present several examples to illustrate the proposed analytical framework. Our focus is on how to obtain a compact expression of the joint MGFs, which can be greatly simplified with the application of the following function and relations.

\section{Common Functions and Useful Relations}

In this section, we introduce some common functions and their properties. These results will be used to simplify the derivation of  joint MGFs in later sections.

\subsection{Common Functions}

\begin{enumerate}
\item[i)]A mixture of a CDF and an MGF $c\left( {\gamma ,\lambda } \right)$:
\begin{equation} \label{eq:CDF_MGF}
\small c\left( {\gamma ,\lambda } \right) = \int_0^\gamma  {dx\;p\left( x \right)\exp \left( {\lambda x} \right)},
\end{equation}
where $p\left( x \right)$ denotes the PDF of the RV of interest.
Note that $c\left( {\gamma , 0 } \right)=c\left( {\gamma } \right)$ is the CDF and $c\left( {\infty ,\lambda } \right)$ leads to the MGF. 
Here, the variable $\gamma$ is real, while $\lambda$ can be complex.

\item[ii)]A mixture of an exceedance distribution function (EDF) and an MGF, $e\left( {\gamma ,\lambda } \right)$:
\begin{equation} \label{eq:EDF_MGF}
\small e\left( {\gamma ,\lambda } \right) = \int_\gamma ^\infty  {dx\;p\left( x \right)\exp \left( {\lambda x} \right)}.
\end{equation}
Note that $e\left( {\gamma ,0} \right)=e\left( {\gamma} \right)$ is the EDF while $e\left( {0,\lambda } \right)$ gives the MGF.

\item[iii)]An Interval MGF $\mu \left( {\gamma _a ,\gamma _b ,\lambda } \right)$:
\begin{equation} \label{eq:IntervalMGF}
\small \mu \left( {\gamma _a ,\gamma _b ,\lambda } \right) = \int_{\gamma _a }^{\gamma _b } {dx\;p\left( x \right)\exp \left( {\lambda x} \right)}.
\end{equation}
Note that $\mu\left( {0, \infty, \lambda} \right)$ gives the MGF.
\end{enumerate}
Note that the functions defined in (\ref{eq:CDF_MGF}), (\ref{eq:EDF_MGF}) and (\ref{eq:IntervalMGF}) are related as follows
\small \begin{eqnarray}
 c\left( {\gamma ,\lambda } \right) &=& e\left( {0,\lambda } \right) - e\left( {\gamma ,\lambda } \right) \\ \nonumber
  &=& c\left( {\infty ,\lambda } \right) - e\left( {\gamma ,\lambda } \right) \\ 
 e\left( {\gamma ,\lambda } \right) &=& c\left( {\infty ,\lambda } \right) - c\left( {\gamma ,\lambda } \right) \\ \nonumber
  &=& e\left( {0,\lambda } \right) - c\left( {\gamma ,\lambda } \right) \\ 
 \mu \left( {\gamma _a ,\gamma _b ,\lambda } \right) &=& c\left( {\gamma _b ,\lambda } \right) - c\left( {\gamma _a ,\lambda } \right) \\ \nonumber
  &=& e\left( {\gamma _a ,\lambda } \right) - e\left( {\gamma _b ,\lambda } \right). \label{eq:Interrelation_of_Interval_MGF}
\end{eqnarray} \normalsize
\subsection{Simplifying Relationship}
\begin{enumerate}
\item[i)] Integral $I_m$:
\\
Based on the derivation given Appendix~\ref{AP:A}, the integral $I_m$ defined as:
\small
\begin{eqnarray}
 I_m \!\!\!\! &=& \!\!\!\! \int_0^{\gamma _{m - 1:K} } \!\!\!\! {d\gamma _{m:K} \;p\left( {\gamma _{m:K} } \right)\exp \left( {\lambda \gamma _{m:K} } \right)\int_0^{\gamma _{m:K} } \!\!\!\! {d\gamma _{m + 1:K} p\left( {\gamma _{m + 1:K} } \right)\exp \left( {\lambda \gamma _{m + 1:K} } \right)} }  \nonumber \\ 
\!\!\!\! && \!\!\!\! \times \int_0^{\gamma _{m + 1:K} }\!\!\!\! {d\gamma _{m + 2:K} p\left( {\gamma _{m + 2:K} } \right)\exp \left( {\lambda \gamma _{m + 2:K} } \right)}  \cdots \int_0^{\gamma _{K - 1:K} } \!\!\!\! {d\gamma _{K:K} p\left( {\gamma _{K:K} } \right)\exp \left( {\lambda \gamma _{K:K} } \right)}, 
\end{eqnarray} \normalsize
can be expressed in terms of the function $c\left(\gamma, \lambda\right)$ as
\begin{equation}  \label{eq:CDF_MGF_multiple}\small
I_m=\frac{1}{{\left( {K - m + 1} \right)!}}\left[ {c\left( {\gamma _{m - 1:K} ,\lambda } \right)} \right]^{\left( {K - m + 1} \right)}.
\end{equation}
\item[ii)]Integral $I'_m$:
\\
Following the similar derivation as given in Appendix~\ref{AP:A}, the integral $I'_m$ defined as
\small\begin{eqnarray} 
 I'_m \!\!\!\! &=& \!\!\!\! \int_{\gamma _{m + 1:K} }^\infty \!\!\!\! {d\gamma _{m:K} \;p\left( {\gamma _{m:K} } \right)\exp \left( {\lambda \gamma _{m:K} } \right)\int_{\gamma _{m:K} }^\infty \!\!\!\! {d\gamma _{m - 1:K} p\left( {\gamma _{m - 1:K} } \right)\exp \left( {\lambda \gamma _{m - 1:K} } \right)} } \nonumber \\ 
\!\!\!\! && \!\!\!\! \times \int_{\gamma _{m - 1:K} }^\infty \!\!\!\! {d\gamma _{m - 2:K} p\left( {\gamma _{m - 2:K} } \right)\exp \left( {\lambda \gamma _{m - 2:K} } \right)}  \cdots \int_{\gamma _{2:K} }^\infty \!\!\!\! {d\gamma _{1:K} p\left( {\gamma _{1:K} } \right)\exp \left( {\lambda \gamma _{1:K} } \right)}, 
\end{eqnarray} \normalsize
can be expressed in terms of the function $e\left(\gamma, \lambda\right)$ as
\begin{equation}\small\label{eq:EDF_MGF_multiple}
I'_m= \frac{1}{{m!}}\left[ {e\left( {\gamma _{m + 1:K} ,\lambda } \right)} \right]^m.
\end{equation}

\item[iii)] Integral $I''_{a,b}$:
\\
Based on the derivation given in Appendix~\ref{AP:C}, the integral $I''_{a,b}$ defined as
\small \begin{eqnarray} 
 I''_{a,b} \!\!\!\! &=& \!\!\!\! \int_{\gamma _{b:K} }^{\gamma _{a:K} }\!\!\!\! {d\gamma _{b - 1:K} \;p\left( {\gamma _{b - 1:K} } \right)\exp \left( {\lambda \gamma _{b - 1:K} } \right)\int_{\gamma _{b - 1:K} }^{\gamma _{a:K} } \!\!\!\! {d\gamma _{b - 2:K} p\left( {\gamma _{b - 2:K} } \right)\exp \left( {\lambda \gamma _{b - 2:K} } \right)} }  \nonumber \\
\!\!\!\!  && \!\!\!\! \times \int_{\gamma _{b - 2:K} }^{\gamma _{a:K} } \!\!\!\! {d\gamma _{b - 3:K} p\left( {\gamma _{b - 3:K} } \right)\exp \left( {\lambda \gamma _{b - 3:K} } \right)}  \cdots \int_{\gamma _{a + 2:K} }^{\gamma _{a:K} } \!\!\!\! {d\gamma _{a+1:K} p\left( {\gamma _{a+1:K} } \right)\exp \left( {\lambda \gamma _{a+1:K} } \right)},
\end{eqnarray} \normalsize
can be expressed in terms of the function $\mu\left(\cdot,\cdot\right)$ as
\begin{equation} \small\label{eq:IntervalMGF_multiple}
I''_{a,b}= \frac{1}{{\left( {b - a - 1} \right)!}}\left[ {\mu \left( {\gamma _{b:K} ,\gamma _{a:K} ,\lambda } \right)} \right]^{\left( {b - a - 1} \right)} \quad\quad\text{for }b>a.
\end{equation}
\end{enumerate}


\section{Sample Cases when All $K$ Ordered RVs are Considered}
Assume the original RVs $\left\{ \gamma_i \right\}$ are i.i.d. with a common arbitrary PDF $p \left( \gamma \right)$, the $K$-dimensional joint PDF of $\left\{ {\gamma _{i:K} } \right\}_{i = 1}^K$ is simply given by~\cite{kn:Order_Statistics}
\begin{equation} \label{eq:m-joint_PDF_MRC}
\small p\left( {\gamma _{1:K} ,\gamma _{2:K} , \cdots ,\gamma _{K:K} } \right) = F \cdot \prod\limits_{i = 1}^K {p\left( {\gamma _{i:K} } \right)} \quad \text{for} \;\; \gamma_{1:K} \ge \gamma_{2:K} \ge \gamma_{3:K} \cdots \ge \gamma_{K:K},
\end{equation}
where $F=K!$.

\begin{theorem} (PDF of $\sum\limits_{n = 1}^K {\gamma _{n:K} }$ among $K$ ordered RVs)\\
Let $Z_1 = \sum\limits_{n=1}^{K} {\gamma_{n:K}}$ for convenience. We can derive the PDF of $Z=\left[Z_1\right]$ as
\small \begin{eqnarray} p_Z \left( {z_1 } \right) 
&=& \mathcal{L}_{S_1 }^{ - 1} \left\{ {\left[ {c\left( {\infty ,-S _1 } \right)} \right]^K } \right\},
\end{eqnarray} \normalsize
where $\mathcal{L}_{S_1 }^{ - 1}\{\cdot\}$ denotes the inverse Laplace transform with respect to $S_1$. 
\end{theorem}
\begin{proof}
The MGF of $Z$ is given by the expectation
\small \begin{eqnarray} \label{eq:MGF_of_pure_MRC_integralform} \!\!\!\!\!\!\!\!\!\!\!\!\!\!\!\! \!\!\!\!\!\!\!\! MGF_Z \left( {\lambda _1 } \right) = {\rm E}\left\{ {\exp \left( {\lambda _1 Z_1 } \right)} \right\} &=& \!\!F \cdot \!\!\! \int\limits_0^\infty  \!\! {d\gamma _{1:K} p\left( {\gamma _{1:K} } \right)\exp \left( {\lambda _1 \gamma _{1:K} } \right) \!\!\!\! \int\limits_0^{\gamma _{1:K} } \!\!\!\! {d\gamma _{2:K} p\left( {\gamma _{2:K} } \right)\exp \left( {\lambda _1 \gamma _{2:K} } \right)} }\nonumber\\ \!\!\!\!\!\!\!\!\!\!\!\!\!\!\!\! \!\!\!\!\!\!\!\!    && \!\! \times \cdots \times \!\!\!\!\!\!\!\!  \int\limits_0^{\gamma _{K - 1:K} } \!\!\!\!\!\!\! {d\gamma _{K:K} p\left( {\gamma _{K:K} } \right)\exp \left( {\lambda _1 \gamma _{K:K} } \right)} ,\end{eqnarray} \normalsize where ${\rm E}\left\{ \cdot \right\}$ denotes the expectation operator.
By applying (\ref{eq:CDF_MGF_multiple}), we can obtain the MGF of $Z_1 = \sum_{m=1}^{K} {\gamma_{m:K}}$ as
\small \begin{eqnarray} \label{eq:MGF_of_pure_MRC}MGF_Z \left( {\lambda _1 } \right) 
&=& \left[ {c\left( {\infty ,\lambda _1 } \right)} \right]^K.\end{eqnarray} \normalsize
Therefore, we can derive the PDF of $Z_1 = \sum_{m=1}^{K} {\gamma_{m:K}}$ by applying the inverse Laplace transform as
\small \begin{eqnarray} p_Z \left( {z_1 } \right) 
&=& \mathcal{L}_{S_1 }^{ - 1} \left\{ {MGF_Z \left( {-S _1 } \right)} \right\} \nonumber \\   &=& \mathcal{L}_{S_1 }^{ - 1} \left\{ {\left[ {c\left( {\infty ,-S _1 } \right)} \right]^K } \right\}.
\end{eqnarray} \normalsize
\end{proof}

\begin{theorem} (Joint PDF of $\gamma _{m:K}$ and $\sum\limits_{\scriptstyle n = 1 \hfill \atop \scriptstyle n \ne m \hfill}^K {\gamma _{n:K} } $)
\\
Let $Z_1  = \gamma _{m:K}$ and $Z_2  = \sum\limits_{\scriptstyle n = 1 \hfill \atop \scriptstyle n \ne m \hfill}^K {\gamma _{n:K} }$ for convenience. We can obtain the 2-dimensional joint PDF of $Z=\left[Z_1, Z_2 \right]$ as
\small \begin{eqnarray}
\!\!\!\!\!\!\!\!\!\!\!\!\!\!\!\!\!\!\!\!\!\!\!\!\!\!\!\! && \!\!\!\! p_Z \left( {z_1 ,z_2 } \right) = p_{\gamma _{m:K},\sum\limits_{\scriptstyle n = 1 \hfill \atop \scriptstyle n \ne m \hfill}^K {\gamma _{n:K} }} \left( {z_1 ,z_2 } \right)\nonumber
\\
\!\!\!\!\!\!\!\!\!\!\!\!\!\!\!\!\!\!\!\!\!\!\!\!\!\!\!\! &=& \!\! \begin{cases}
\!\! \frac{{K!}}{{\left( {K - 1} \right)!}}p\left( {z_1 } \right)\mathcal{L}_{S_2 }^{ - 1} \left\{ {\left[ {c\left( {z_1 , - S_2 } \right)} \right]^{\left( {K - 1} \right)} } \right\}&{\rm for }\; m=1, \; z_1\ge \frac{1}{K-1}z_2
\\
\!\! \frac{{K!}}{{\left( {K - m} \right)!\left( {m - 1} \right)!}}p\left( {z_1 } \right)\mathcal{L}_{S_2 }^{ - 1} \left\{ {\left[ {c\left( {z_1 , - S_2 } \right)} \right]^{\left( {K - m} \right)} \left[ {e\left( {z_1 , - S_2 } \right)} \right]^{\left( {m - 1} \right)} } \right\} &{\rm for }\;m\ge2.
\end{cases}
\end{eqnarray} \normalsize\:
\end{theorem}

\begin{proof}
The second order MGF of $Z=\left[Z_1,Z_2\right]$ is given by the expectation
\small \begin{eqnarray} \label{eq:joint_MGF_2_integralform}
 \!\!\!\!\!\!\!\!\!\!\!\!\!\!\!\! && \!\!\!\!\!\!\!\!MGF_Z \left( {\lambda _1 ,\lambda _2 } \right) = E\left\{ {\exp \left( {\lambda _1 Z_1  + \lambda _2 Z_2 } \right)} \right\} \nonumber
\\ 
 \!\!\!\!\!\!\!\!\!\!\!\!\!\!\!\! &=& \!\! F\!\cdot\!\!\! \int\limits_0^\infty \!\! {d\gamma _{1:K} p\left( {\gamma _{1:K} } \right)\exp \left( {\lambda _2 \gamma _{1:K} } \right) \!\!\!\! \int\limits_0^{\gamma _{1:K} } \!\!\!\! {d\gamma _{2:K} p\left( {\gamma _{2:K} } \right)\exp \left( {\lambda _2 \gamma _{2:K} } \right)}  \cdots \!\!\!\!\!\!\!\! \int\limits_0^{\gamma _{m - 2:K} }  \!\!\!\!\!\!\! {d\gamma _{m - 1:K} p\left( {\gamma _{m - 1:K} } \right)\exp \left( {\lambda _2 \gamma _{m - 1:K} } \right)} } \nonumber
\\ 
 \!\!\!\!\!\!\!\!\!\!\!\!\!\!\!\! &&  \times \!\!\!\!\!\! \int\limits_0^{\gamma _{m - 1:K} } {d\gamma _{m:K} p\left( {\gamma _{m:K} } \right)\exp \left( {\lambda _1 \gamma _{m:K} } \right)} \nonumber 
\\ 
 \!\!\!\!\!\!\!\!\!\!\!\!\!\!\!\! &&  \times \!\!\!\!\! \int\limits_0^{\gamma _{m:K} } {d\gamma _{m + 1:K} p\left( {\gamma _{m + 1:K} } \right)\exp \left( {\lambda _2 \gamma _{m + 1:K} } \right)}  \cdots \int\limits_0^{\gamma _{K - 1:K} } {d\gamma _{K:K} p\left( {\gamma _{K:K} } \right)\exp \left( {\lambda _2 \gamma _{K:K} } \right)}. 
\end{eqnarray} \normalsize
We show in Appendix~\ref{AP:D} that by applying (\ref{eq:CDF_MGF_multiple}), (\ref{eq:Integral_solution}) and (\ref{eq:EDF_MGF_multiple}), we can obtain the second order MGF of $Z_1  = \gamma _{m:K}$ and $Z_2  = \sum\limits_{\scriptstyle n = 1 \hfill \atop \scriptstyle n \ne m \hfill}^K {\gamma _{n:K} }$ as
\small \begin{eqnarray} \label{eq:joint_MGF_1}
 \!\!\!\!\!\!\!\!\!\!\!\!\!\!\!\!\!\!\!\!\!\!\!\!\!\!\!\!\!\! && \!\!\!\!\!\!\!\!MGF_Z \left( {\lambda _1 ,\lambda _2 } \right) \nonumber
\\
 \!\!\!\!\!\!\!\!\!\!\!\!\!\!\!\!\!\!\!\!\!\!\!\!\!\!\!\!\!\! &=& \!\! \frac{F}{{\left( {K - m} \right)!\left( {m - 1} \right)!}}\int\limits_0^\infty  {d\gamma _{m:K} p\left( {\gamma _{m:K} } \right)\exp \left( {\lambda _1 \gamma _{m:K} } \right)\left[ {c\left( {\gamma _{m:K} ,\lambda _2 } \right)} \right]^{\left( {K - m} \right)} } \left[ {e\left( {\gamma _{m:K} ,\lambda _2 } \right)} \right]^{\left( {m - 1} \right)}. 
\end{eqnarray} \normalsize

With the MGF expression given in (\ref{eq:joint_MGF_1}) at hand, we are now in the position to derive the 2-dimensional joint PDF of $Z_1  = \gamma _{m:K}$ and $Z_2  = \sum\limits_{\scriptstyle n = 1 \hfill \atop \scriptstyle n \ne m \hfill}^K {\gamma _{n:K} }$. Letting $\lambda _1  =  - S_1$ and $\lambda _2  =  - S_2$, we can derive the 2-dimensional joint PDF by applying the inverse Laplace transform as
\small \begin{eqnarray} \label{eq:joint_PDF_2_u}
 \!\!\!\!\!\!\!\!\!\!\!\! && \!\!\!\!\!\!\!\! p_Z \left( {z_1 ,z_2 } \right) = \mathcal{L}_{S_1 ,S_2 }^{ - 1} \left\{ {MGF_Z \left( { - S_1 , - S_2 } \right)} \right\} \nonumber
 \\ 
 \!\!\!\!\!\!\!\!\!\!\!\!&=& \!\!  \frac{{K!}}{{\left( {K - m} \right)!\left( {m - 1} \right)!}}\int\limits_0^\infty  {d\gamma _{m:K} \Bigg[p\left( {\gamma _{m:K} } \right)\mathcal{L}_{S_1 }^{ - 1} \left\{ {\exp \left( { - S_1 \gamma _{m:K} } \right)} \right\}} \nonumber
\\
 \!\!\!\!\!\!\!\!\!\!\!\! && \!\!  \times \mathcal{L}_{S_2 }^{ - 1} \left\{ {\left[ {c\left( {\gamma _{m:K} , - S_2 } \right)} \right]^{\left( {K - m} \right)} \left[ {e\left( {\gamma _{m:K} , - S_2 } \right)} \right]^{\left( {m - 1} \right)} } \right\}\Bigg].
\end{eqnarray} \normalsize
Based on the inverse Laplace transform properties in Appendix~\ref{AP:Inverse_LT},
\begin{equation} \label{eq:inverse_LT_1}
\small \mathcal{L}_{S_1 }^{ - 1} \left\{ {\exp \left( { - S_1 \gamma _{m:K} } \right)} \right\} = \delta \left( {z_1  - \gamma _{m:K} } \right).
\end{equation}
Therefore, substituting (\ref{eq:inverse_LT_1}) in (\ref{eq:joint_PDF_2_u}) we can obtain the desired 2-dimensional joint PDF of $Z_1  = \gamma _{m:K}$ and $Z_2  = \sum\limits_{\scriptstyle n = 1 \hfill \atop \scriptstyle n \ne m \hfill}^K {\gamma _{n:K} }$.
\normalsize
\end{proof}

\begin{theorem} ({Joint PDF of $\sum\limits_{n = 1}^m {\gamma _{n:K} }$ and $\sum\limits_{n = m + 1}^K {\gamma _{n:K} }$})
\\
Let $Z_1  = \sum\limits_{n = 1}^m {\gamma _{n:K} }$ and $Z_2  = \sum\limits_{n = m + 1}^K {\gamma _{n:K} }$ for convenience, then we can derive the 2-dimensional joint PDF of $Z=\left[Z_1, Z_2 \right]$  as
\small \begin{eqnarray} \label{eq:joint_PDF_3}
\!\!\!\!\!\!\!\!\!\!\!\! p_Z \left( {z_1 ,z_2 } \right) \!\! &=& \!\! p_{\sum\limits_{n = 1}^m {\gamma _{n:K} },\sum\limits_{n = m + 1}^K {\gamma _{n:K} }} \left( {z_1 ,z_2 } \right) \nonumber
\\
\!\! &=& \!\! \mathcal{L}_{S_1 ,S_2 }^{ - 1} \left\{ {MGF_Z \left( { - S_1 , - S_2 } \right)} \right\} \nonumber
\\ 
\!\! &=& \!\!  \frac{{K!}}{{\left( {K - m} \right)!\left( {m - 1} \right)!}}  \int\limits_0^\infty  {d\gamma _{m:K} \Bigg[p\left( {\gamma _{m:K} } \right) \mathcal{L}_{S_1}^{-1} \left\{ {\exp \left( {-S _1 \gamma _{m:K} } \right)\left[ {e\left( {\gamma _{m:K} ,-S _1 } \right)} \right]^{\left( {m - 1} \right)} } \right\}} \nonumber
\\
\!\! &&\!\!  \times \mathcal{L}_{S_2}^{-1} \left\{ {\left[ {c\left( {\gamma _{m:K} ,-S _2 } \right)} \right]^{\left( {K - m} \right)} } \right\}\Bigg]\quad \quad {\rm for}\; z_1\ge \frac{m}{K-m}z_2.
\end{eqnarray} \normalsize
\end{theorem}
\begin{proof}
The second order MGF of $Z=\left[Z_1,Z_2\right]$ is given by the expectation
\small \begin{eqnarray} \label{eq:joint_MGF_7_integralform}
 \!\!\!\!\!\!\!\!\!\!\!\! && \!\!\!\!\!\!\!\! MGF_Z \left( {\lambda _1 ,\lambda _2 } \right) = E\left\{ {\exp \left( {\lambda _1 Z_1  + \lambda _2 Z_2 } \right)} \right\} \nonumber
\\ 
 \!\!\!\!\!\!\!\!\!\!\!\!&=& \!\! F\int\limits_0^\infty  {d\gamma _{1:K} p\left( {\gamma _{1:K} } \right)\exp \left( {\lambda _1 \gamma _{1:K} } \right) \cdots \int\limits_0^{\gamma _{m - 1:K} } {d\gamma _{m:K} p\left( {\gamma _{m:K} } \right)\exp \left( {\lambda _1 \gamma _{m:K} } \right)} } \\ 
 \!\!\!\!\!\!\!\!\!\!\!\! &&\times \!\! \int\limits_0^{\gamma _{m:K} } {d\gamma _{m + 1:K} p\left( {\gamma _{m + 1:K} } \right)\exp \left( {\lambda _2 \gamma _{m + 1:K} } \right)}  \cdots \int\limits_0^{\gamma _{K - 1:K} } {d\gamma _{K:K} p\left( {\gamma _{K:K} } \right)\exp \left( {\lambda _2 \gamma _{K:K} } \right)}.
\end{eqnarray} \normalsize

\noindent We show in Appendix~\ref{AP:E} that by applying (\ref{eq:CDF_MGF_multiple}) and (\ref{eq:Integral_solution}) and then (\ref{eq:EDF_MGF_multiple}), we can obtain the second order MGF of $Z$ as
\small \begin{eqnarray} \label{eq:joint_MGF_2}
\!\!\!\!\!\!\!\!\!\!\!\!\!\!\!\!\!\!\!\!\!\!\!\!\!\!\!\! && \!\!\!\!\!\!\!\! MGF_Z \left( {\lambda _1 ,\lambda _2 } \right) \nonumber
\\
 \!\!\!\!\!\!\!\!\!\!\!\!\!\!\!\!\!\!\!\!\!\!\!\!\!\!\!\! &=& \!\! \frac{{K!}}{{\left( {K - m} \right)!\left( {m - 1} \right)!}}\int\limits_0^\infty  {d\gamma _{m:K} p\left( {\gamma _{m:K} } \right)\exp \left( {\lambda _1 \gamma _{m:K} } \right)\left[ {c\left( {\gamma _{m:K} ,\lambda _2 } \right)} \right]^{\left( {K - m} \right)} \left[ {e\left( {\gamma _{m:K} ,\lambda _1 } \right)} \right]^{\left( {m - 1} \right)} }.
\end{eqnarray} \normalsize

Again, letting $\lambda _1  =  - S_1$ and $\lambda _2  =  - S_2$, we can obtain the desired 2-dimensional joint PDF of $Z_1  = \sum\limits_{n = 1}^m {\gamma _{n:K} }$ and $Z_2  = \sum\limits_{n = m + 1}^K {\gamma _{n:K} }$ by applying the inverse Laplace transform.
\normalsize
\end{proof}

\section{Sample cases when only $K_s$ ordered RVs are considered}
Let us now consider the cases where only the best $K_s$($\le K$) ordered RVs are involved. Assuming the original $\left\{ \gamma_i \right\}$ are  i.i.d. RVs with a common arbitrary PDF $p \left( \gamma \right)$ and CDF $P\left( {\gamma } \right)$, the $K_s$-dimensional joint PDF of $\left\{ {\gamma _{i:K} } \right\}_{i = 1}^{K_s}$ is simply given by~\cite{kn:Order_Statistics}
\begin{equation} \label{eq:m-joint_PDF_GSC}
\small p\left( {\gamma _{1:K} ,\gamma _{2:K} , \cdots ,\gamma _{K_s :K} } \right) = F \cdot \prod\limits_{i = 1}^{K_s } {p\left( {\gamma _{i:K} } \right)} \left[ {P\left( {\gamma _{K_s :K} } \right)} \right]^{K - K_s }, 
\end{equation}
where $F = K_s !\binom{K}{K_s} = \frac{K!}{\left( {K - K_s } \right)!}$.

\begin{theorem} ({PDF of $\sum\limits_{n = 1}^{K_s } {\gamma _{n:K} }$}, $K_s\ge 2$)
\\
Let $Z' = \sum\limits_{n=1}^{K_s} {\gamma_{n:K}}$ for convenience, then we can derive the PDF of $Z'$ as
\small \begin{eqnarray} \label{eq:PDF_of_pure_GSC}
p_{Z'} \left( x \right) = p_{\sum_{n = 1}^{K_s } {\gamma _{n:K} } } \left( x \right)=
\int_0^{\frac{x}{{K_s }}} {p_Z \left( {x - z_2 ,z_2 } \right)dz_2 } & \text{for } K_s \ge 2, \end{eqnarray} \normalsize
\end{theorem}
where
\begin{equation} \label{eq:PDF_of_pure_GSC_m} \small
p_Z \left( {z_1 ,z_2 } \right) = \frac{F}{{\left( {K_s  - 1} \right)!}}p\left( {z_2 } \right)\left[ {c\left( {z_2 } \right)} \right]^{\left( {K - K_s } \right)} \mathcal{L}_{S_1 }^{ - 1}\left\{ {\left[ {e\left( {z_2 , - S_1 } \right)} \right]^{\left( {K_s  - 1} \right)} } \right\}.
\end{equation}

\begin{proof}
We only need to consider $\gamma_{K_s:K}$ separately in this case. 
Let $Z_1=\sum\limits_{n=1}^{K_s-1} {\gamma_{n:K}}$ and $Z_2=\gamma_{K_s:K}$. The target second order MGF of $Z=[Z_1, Z_2]$ is given by the expectation in \small \begin{eqnarray} \label{eq:MGF_of_pure_GSC_integralform}
MGF_Z \left( {\lambda _{1,} \lambda _2 } \right) &=& {\rm E}\left\{ {\exp \left( {\lambda _1 z_1  + \lambda _2 z_2 } \right)} \right\} \nonumber
\\
&=& F\int_0^\infty  {d\gamma _{1:K} p\left( {\gamma _{1:K} } \right)\exp \left( {\lambda _1 \gamma _{1:K} } \right) \cdots \int_0^{\gamma _{K_s  - 2:K} } {d\gamma _{K_s  - 1:K} } } p\left( {\gamma _{K_s  - 1:K} } \right)\exp \left( {\lambda _1 \gamma _{K_s  - 1:K} } \right) \nonumber
\\
&& \times \int_0^{_{K_s  - 1:K} } {d\gamma _{K_s :K} p\left( {\gamma _{K_s :K} } \right)\exp \left( {\lambda _2 \gamma _{K_s :K} } \right)\left[ {c\left( {\gamma _{K_s :K} } \right)} \right]^{K - K_s } }.
\end{eqnarray} \normalsize
By simply applying (\ref{eq:Integral_solution}) and then (\ref{eq:EDF_MGF_multiple}) to (\ref{eq:MGF_of_pure_GSC_integralform}), we can obtain the second order MGF result as
\small\begin{eqnarray} \label{MGF_pure_GSC}
\!\!\!\!MGF_Z \left( {\lambda _{1,} \lambda _2 } \right) = \frac{F}{{\left( {K_s  - 1} \right)!}}\int_0^\infty  {d\gamma _{K_s :K} p\left( {\gamma _{K_s :K} } \right)\exp \left( {\lambda _2 \gamma _{K_s :K} } \right)\left[ {c\left( {\gamma _{K_s :K} } \right)} \right]^{K - K_s } } \left[ {e\left( {\gamma _{K_s :K} },\lambda_1 \right)} \right]^{K_s  - 1}.
\end{eqnarray} \normalsize
Again, letting $\lambda _1  =  - S_1$ and $\lambda _2  =  - S_2$, we can obtain the 2-dimensional joint PDF of $Z_1=\sum\limits_{n=1}^{K_s-1} {\gamma_{n:K}}$ and $Z_2=\gamma_{K_s:K}$ by applying the inverse Laplace transform. 
Finally, noting that  $Z'=Z_1+Z_2$, we can obtain the target PDF of $Z'$ with the following finite integration
\begin{equation} \small
p_{Z'}(x)=\int_0^{\frac{x}{{K_s }}} {p_Z \left( {x - z_2 ,z_2 } \right)dz_2 }.
\end{equation}\normalsize
\end{proof}

\begin{theorem} ({Joint PDF of $\gamma _{m:K}$ and $\sum\limits_{\scriptstyle n = 1 \hfill \atop \scriptstyle n \ne m \hfill}^{K_s } {\gamma _{n:K} }$})

Let $X=\gamma _{m:K}$ and $Y=\sum\limits_{\scriptstyle n = 1 \hfill \atop \scriptstyle n \ne m \hfill}^{K_s } {\gamma _{n:K} }$, then the joint PDF of $Z=\left[X, Y\right]$ can be obtained as
\begin{eqnarray} \label{eq:joint_PDF_1}
\!\!\!\!\!\!\!\!\!\!\!\!\!\!  && \!\!\!\!\!\!\! p_{Z}\left(x, y\right) = p_{\gamma _{m:K} ,\sum\limits_{\scriptstyle n = 1 \hfill \atop 
  \scriptstyle n \ne m \hfill}^{K_s } {\gamma _{n:K} } } \left( {x,y} \right) \nonumber
\\
\!\!\!\!\!\!\!\!\!\!\!\!\!\! &=&\!\!\!\!\!\!\!
\begin{cases}
\int_{\left( {\frac{{K_s  - 2}}{{K_s  - 1}}} \right)y}^{\left( {K_s  - 2} \right)x} {p_{\gamma _{1:K} ,\sum\limits_{n = 2}^{K_s  - 1} {\gamma _{n:K} } ,\gamma _{K_s :K} } \left( {x,z_2 ,y - z_2 } \right)dz_2 } , &m = 1,
\\
\int_0^x {\int_{\left( {m - 1} \right)x}^{y - \left( {K_s  - m} \right)z_4 } {p_{\sum\limits_{n = 1}^{m - 1} {\gamma _{n:K} } ,\gamma _{m:K} ,\sum\limits_{n = m + 1}^{K_s  - 1} {\gamma _{n:K} } ,\gamma _{K_s :K} } \!\!\!\!\!\!\! \left( {z_1 ,x,y \!-\!z_1\!-\!z_4 ,z_4 } \right)dz_1 } dz_4 } , &1 < m < K_s  - 1,
\\
\int_{\left( {K_s  - 2} \right)x}^y {p_{\sum\limits_{n = 1}^{K_s  - 2} {\gamma _{n:K} } ,\gamma _{K_s  - 1:K} ,\gamma _{K_s :K} } \left( {z_1 ,x,y - z_1 } \right)dz_1 } , &m = K_s  - 1,
\\
p_{\gamma _{K_s :K} ,\sum\limits_{n = 1}^{K_s  - 1} {\gamma _{n:K} } } \left( {x,y} \right), &m = K_s,
\end{cases}
\end{eqnarray}
\\
or equivalently
 \begin{eqnarray} \label{eq:joint_PDF_2}
\!\!\!\!\!\!\!\!\!\!\!\!\!\!  && \!\!\!\!\!\!\! p_{\gamma _{m:K} ,\sum\limits_{\scriptstyle n = 1 \hfill \atop 
  \scriptstyle n \ne m \hfill}^{K_s } {\gamma _{n:K} } } \left( {x,y} \right) \nonumber
\\
\!\!\!\!\!\!\!\!\!\!\!\!\!\! &=&\!\!\!\!\!\!\!
\begin{cases}
\int_{\left( {\frac{{K_s  - 2}}{{K_s  - 1}}} \right)y}^{\left( {K_s  - 2} \right)x} {p_{\gamma _{1:K} ,\sum\limits_{n = 2}^{K_s  - 1} {\gamma _{n:K} } ,\gamma _{K_s :K} } \left( {x,z_2 ,y - z_2 } \right)dz_2 } , & m = 1,
\\
\int_0^x {\int_{\left( {K_s  - m - 1} \right)z_4 }^{\left( {K_s  - m - 1} \right)x} {p_{\sum\limits_{n = 1}^{m - 1} {\gamma _{n:K} } ,\gamma _{m:K} ,\sum\limits_{n = m + 1}^{K_s  - 1} {\gamma _{n:K} } ,\gamma _{K_s :K} } \!\!\!\!\!\!\! \left( {y \!- \!z_3 \! - \!z_4 ,x,z_3 ,z_4 } \right)dz_3 } dz_4 }, &1 < m < K_s  - 1,
\\
\int_0^x {p_{\sum\limits_{n = 1}^{K_s  - 2} {\gamma _{n:K} } ,\gamma _{K_s  - 1:K} ,\gamma _{K_s :K} } \left( {y - z_3 ,x,z_3 } \right)dz_3 } , &m = K_s  - 1,
\\
p_{\gamma _{K_s :K} ,\sum\limits_{n = 1}^{K_s  - 1} {\gamma _{n:K} } } \left( {x,y} \right), &m = K_s.
\end{cases}
\end{eqnarray}
\end{theorem}
\begin{proof}
To derive the joint PDF of $\gamma _{m:K}$ and $\sum\limits_{\scriptstyle n = 1 \hfill \atop \scriptstyle n \ne m \hfill}^{K_s } {\gamma _{n:K} }$, we need to consider four cases i) $m=1$, ii) $ 1<m<K_s-1$, iii) $m=K_s-1$ and iv) $m=K_s$ separately based on our unified framework. While for case iv), we can start with the second order MGF of $\gamma _{m:K}$ and $\sum\limits_{\scriptstyle n = 1 \hfill \atop \scriptstyle n \ne m \hfill}^{K_s } {\gamma _{n:K} }$ directly, we should consider substituted groups instead of the original groups for cases i), ii), and iii). More specifically, for cases i) and iii), we need to consider $\gamma_{K_s:K}$ separately as shown in Fig.~\ref{Example_3_1} and~\ref{Example_3_3} whereas, for case ii), as one of original groups is split by $\gamma_{m:K}$, we should consider substituted groups for the split group instead of original groups  as shown in Fig.~\ref{Example_3_2}. As a result, we will start by obtaining a four order MGF for case ii) and a three order MGF for case i) and case iii). In all these cases, the higher dimensional joint PDF can then be used to find the desired 2-dimensional joint PDF of interest by transformation.

Applying the results in (\ref{eq:Integral_solution}), (\ref{eq:CDF_MGF_multiple}), (\ref{eq:EDF_MGF_multiple}) and (\ref{eq:IntervalMGF_multiple}), we derive in Appendix~\ref{AP:G} the following joint MGF for different cases
\newpage
\begin{enumerate}
\item[a.] $m=1$ 
\\Let $Z_1  = \gamma _{1:K}$, $Z_2  = \sum\limits_{n = 2}^{K_s  - 1} {\gamma _{n:K} }$, and $Z_3  = \gamma _{K_s :K}
$, then
\small \begin{eqnarray} \label{eq:joint_MGF_GSC_3}
\!\!\!\!\!\!\!\!\!\!\!\! && \!\!\!\!\!\!\!\! MGF_Z \left( {\lambda _1 ,\lambda _2 ,\lambda _3 } \right) \nonumber
\\
\!\!\!\!\!\!\!\!\!\!\!\! &=& \!\!\!\! F\int\limits_0^\infty  {d\gamma _{K_s :K} p\left( {\gamma _{K_s :K} } \right)\exp \left( {\lambda _3 \gamma _{K_s :K} } \right)\left[ {c\left( {\gamma _{K_s :K} } \right)} \right]^{\left( {K - K_s } \right)} }  \nonumber
\\ 
\!\!\!\!\!\!\!\!\!\!\!\! && \!\!\!\!  \times \int\limits_{\gamma _{K_s :K} }^\infty  {d\gamma _{1:K} p\left( {\gamma _{1:K} } \right)\exp \left( {\lambda _1 \gamma _{1:K} } \right)\frac{1}{{\left( {K_s  - 2} \right)!}}\left[ {\mu \left( {\gamma _{K_s :K} ,\gamma _{1:K} ,\lambda _2 } \right)} \right]^{\left( {K_s  - 2} \right)} } . 
\end{eqnarray} \normalsize
\item[b.] $1<m<K_s-1$ 
\\Let $Z_1  = \sum\limits_{n = 1}^{m - 1} {\gamma _{n:K} }$, $Z_2  = \gamma _{m:K}$, $Z_3 = \sum\limits_{n = m + 1}^{K_s  - 1} {\gamma _{n:K} }$, and $Z_4  = \gamma _{K_s :K}$, then
\small \begin{eqnarray} \label{eq:joint_MGF_GSC_2}
\!\!\!\!\!\!\!\!\!\!\!\!&&\!\!\!\!\!\!\!\! MGF_Z \left( {\lambda _1 ,\lambda _2 ,\lambda _3 ,\lambda _4 } \right) \nonumber
\\
\!\!\!\!\!\!\!\!\!\!\!\!&=& \!\!\!\!\frac{F}{{\left( {K_s  - m - 1} \right)!\left( {m - 1} \right)!}}\int\limits_0^\infty  {d\gamma _{K_s :K} p\left( {\gamma _{K_s :K} } \right)\exp \left( {\lambda _4 \gamma _{K_s :K} } \right)\left[ {c\left( {\gamma _{K_s :K} } \right)} \right]^{\left( {K - K_s } \right)} } \nonumber
\\
\!\!\!\!\!\!\!\!\!\!\!\!&&\!\!\!\! \times \!\!\!\! \int\limits_{\gamma _{K_s :K} }^\infty \!\!\! {d\gamma _{m:K} p\left( {\gamma _{m:K} } \right)\exp \left( {\lambda _2 \gamma _{m:K} } \right)\left[ {e\left( {\gamma _{m:K} ,\lambda _1 } \right)} \right]^{\left( {m - 1} \right)} \left[ {\mu \left( {\gamma _{K_s :K} ,\gamma _{m:K} ,\lambda _3 } \right)} \right]^{\left( {K_s  - m - 1} \right)} }.
\end{eqnarray} \normalsize
\item[c.] $m=K_s-1$
\\ Let $Z_1  = \sum\limits_{n = 1}^{K_s  - 2} {\gamma _{n:K} }$, $Z_2  = \gamma _{K_s  - 1:K}$ and $Z_3  = \gamma _{K_s :K}$, then
\small \begin{eqnarray} \label{eq:joint_MGF_GSC_4}
\!\!\!\!\!\!\!\!\!\!\!\!\!\!\!\!&& \!\!\!\!\!\!\!\!MGF_Z \left( {\lambda _1 ,\lambda _2 ,\lambda _3 } \right) \nonumber
\\
\!\!\!\!\!\!\!\!\!\!\!\!&=& \!\!\!\!F\int\limits_0^\infty  {d\gamma _{K_s :K} p\left( {\gamma _{K_s :K} } \right)\exp \left( {\lambda _3 \gamma _{K_s :K} } \right)\left[ {c\left( {\gamma _{K_s :K} } \right)} \right]^{\left( {K - K_s } \right)} }  \nonumber
\\ 
\!\!\!\!\!\!\!\!\!\!\!\! && \!\!\!\!  \times \!\!\!\! \int\limits_{\gamma _{K_s :K} }^\infty \!\!\!\! {d\gamma _{K_s  - 1:K} p\left( {\gamma _{K_s  - 1:K} } \right)\exp \left( {\lambda _2 \gamma _{K_s  - 1:K} } \right)\frac{1}{{\left( {K_s  - 2} \right)!}}\left[ {e\left( {\gamma _{K_s  - 1:K} ,\lambda _1 } \right)} \right]^{\left( {K_s  - 2} \right)} }.
\end{eqnarray} \normalsize
\item[d.] $m=K_s$
\\ Let $Z_1  = \gamma _{K_s :K}$ and $Z_2  = \sum\limits_{n = 1}^{K_s  - 1} {\gamma _{n:K} }$, then
\small \begin{eqnarray} \label{eq:joint_MGF_GSC_5}
 MGF_Z \left( {\lambda _1 ,\lambda _2 } \right) &=& F\int\limits_0^\infty  {d\gamma _{K_s :K} p\left( {\gamma _{K_s :K} } \right)\exp \left( {\lambda _1 \gamma _{K_s :K} } \right)\left[ {c\left( {\gamma _{K_s :K} } \right)} \right]^{\left( {K - K_s } \right)} }  \nonumber
\\ 
 && \times \frac{1}{{\left( {K_s  - 1} \right)!}}\left[ {e\left( {\gamma _{K_s :K} ,\lambda _2 } \right)} \right]^{\left( {K_s  - 1} \right)} .
\end{eqnarray} \normalsize
\end{enumerate}
Starting from the MGF expressions given above, we apply inverse Laplace transforms in Appendix~\ref{AP:G} 
in order to derive the following joint PDFs 
\begin{enumerate}
\item[a.] $m=1$ 
\small \begin{eqnarray}  \label{eq:joint_PDF_GSC_3}
\!\!\!\!\!\!\!\!\!\!\!\! && \!\!\!\!\!\!\!\! p_Z \left( {z_1 ,z_2 ,z_3 } \right) = p_{\gamma _{1:K},\sum\limits_{n = 2}^{K_s  - 1} {\gamma _{n:K} },\gamma _{K_s :K}} \left( {z_1 ,z_2 ,z_3 } \right) \nonumber
\\
\!\!\!\!\!\!\!\!\!\!\!\! &=& \!\!\!\! \frac{F}{{\left( {K_s  - 2} \right)!}}p\left( {z_1 } \right)p\left( {z_3 } \right)\left[ {c\left( {z_3 } \right)} \right]^{\left( {K - K_s } \right)} U\left( {z_1  - z_3 } \right)\mathcal{L}_{S_2 }^{ - 1} \left\{ {\left[ {\mu \left( {z_3 ,z_1 , - S_2 } \right)} \right]^{\left( {K_s  - 2} \right)} } \right\} ,\nonumber
\\
&&\text{for } z_3<z_1, \; \left(K_s-2\right)z_3<z_2<\left(K_s-2\right)z_1.
\end{eqnarray} \normalsize
where $U\left(\cdot\right)$ is the unit step function.
\item[b.] $1<m<K_s-1$
\small \begin{eqnarray} \label{eq:joint_PDF_GSC_2} 
\!\!\!\!\!\!\!\!\!\!\!\!\!\!\!\! && \!\!\!\!\!\!\!\! p_Z \left( {z_1 ,z_2 ,z_3 ,z_4 } \right) =p_{\sum\limits_{n = 1}^{m - 1} {\gamma _{n:K} },\gamma _{m:K},\sum\limits_{n = m + 1}^{K_s  - 1} {\gamma _{n:K} },\gamma _{K_s :K}} \left( {z_1 ,z_2 ,z_3 ,z_4 } \right)\nonumber
\\
\!\!\!\!\!\!\!\!\!\!\!\!\!\!\!\! &=& \!\!\!\! \frac{F}{{\left( {K_s  - m - 1} \right)!\left( {m - 1} \right)!}}p\left( {z_2 } \right)p\left( {z_4 } \right)\left[ {c\left( {z_4 } \right)} \right]^{\left( {K - K_s } \right)} U\left( {z_2  - z_4 } \right) \nonumber
\\ 
\!\!\!\!\!\!\!\!\!\!\!\!\!\!\!\! && \!\!\!\! \times \mathcal{L}_{S_1 }^{ - 1} \left\{ {\left[ {e\left( {z_2 , - S_1 } \right)} \right]^{\left( {m - 1} \right)} } \right\}\mathcal{L}_{S_3 }^{ - 1} \left\{ {\left[ {\mu \left( {z_4 ,z_2 , - S_3 } \right)} \right]^{\left( {K_s  - m - 1} \right)} } \right\} , \nonumber
\\
&&\text{for } z_4<z_2,\; z_1>(m-1)z_2\;\text{and}\; \left(K_s-m-1\right)z_4<z_3<\left(K_s-m-1\right)z_2. 
\end{eqnarray} \normalsize
\item[c.] $m=K_s-1$
\small \begin{eqnarray}  \label{eq:joint_PDF_GSC_4}
\!\!\!\!\!\!\!\!\!\!\!\! && \!\!\!\! \!\!\!\! p_Z \left( {z_1 ,z_2 ,z_3 } \right) = p_{\sum\limits_{n = 1}^{K_s  - 2} {\gamma _{n:K} },\gamma _{K_s  - 1:K},\gamma _{K_s :K}} \left( {z_1 ,z_2 ,z_3 } \right)\nonumber
\\ 
\!\!\!\!\!\!\!\!\!\!\!\! &=& \!\!\!\! \frac{F}{{\left( {K_s  - 2} \right)!}}p\left( {z_2 } \right)p\left( {z_3 } \right)\left[ {c\left( {z_3 } \right)} \right]^{\left( {K - K_s } \right)} U\left( {z_2  - z_3 } \right)\mathcal{L}_{S_1 }^{ - 1} \left\{ {\left[ {e\left( {z_2 , - S_1 } \right)} \right]^{\left( {K_s  - 2} \right)} } \right\}, \nonumber
\\
&&\text{for } z_3<z_2, \; z_1>\left(K_s-2\right)z_2.
\end{eqnarray} \normalsize
\item[d.] $m=K_s$
\small \begin{eqnarray}  \label{eq:joint_PDF_GSC_5}
\!\!\!\!\!\!\!\!\!\!\!\!\!\!\!\!\! && \!\!\!\!\!\!\!\! p_Z \left( {z_1 ,z_2 } \right) =p_{\gamma _{K_s :K},\sum\limits_{n = 1}^{K_s  - 1} {\gamma _{n:K} }} \left( {z_1 ,z_2 } \right) \nonumber
\\ 
\!\!\!\!\!\!\!\!\!\!\!\!\!\!\!\!\! &=& \!\!\!\!\frac{F}{{\left( {K_s  - 1} \right)!}}p\left( {z_1 } \right)\left[ {c\left( {z_1 } \right)} \right]^{\left( {K - K_s } \right)} \mathcal{L}_{S_2 }^{ - 1} \left\{ {\left[ {e\left( {z_1 , - S_2 } \right)} \right]^{\left( {K_s  - 1} \right)} } \right\}, \nonumber
\\
&&\text{for } z_2 \ge \left(K_s-1\right)z_1. 
\end{eqnarray} \normalsize
\end{enumerate}
Finally, the joint PDF of $\gamma _{m:K}$ and ${\sum\limits_{\scriptstyle n = 1 \hfill \atop \scriptstyle n \ne m \hfill}^{K_s } {\gamma _{n:K} } }$ can be obtained.
\normalsize
\end{proof}
Note that (\ref{eq:joint_PDF_1}) and (\ref{eq:joint_PDF_2}) involve only finite integrations of joint PDFs. Therefore, while a generic closed-form expression is not possible, the desired joint PDF can be easily numerically evaluated with the help of integral tables~\cite{kn:Mathematical_handbook, kn:gradshteyn_6} or using standard mathematical packages such as Mathematica or Matlab etc.

\begin{theorem} (Joint PDF of $\sum\limits_{n = 1}^m {\gamma _{n:K} }$ and $\sum\limits_{n = m + 1}^{K_s } {\gamma _{n:K} }$)
\\
Let $X=\sum\limits_{n = 1}^m {\gamma _{n:K} }$ and $Y=\sum\limits_{n = m + 1}^{K_s } {\gamma _{n:K} }$, then we can simply obtain the joint PDF of $ Z=[X,Y]$ as
\small\begin{eqnarray}
 \!\!\!\! \!\!\!\!\!\!\!\!\!\!\!\! \!\!\!\!  && \!\!\!\!\!\!\!\! p_Z \left( {x,y} \right) = p_{\sum\limits_{n = 1}^m {\gamma _{n:K} },\sum\limits_{n = m + 1}^{K_s } {\gamma _{n:K} }} \left( {x,y} \right)\nonumber
\\
\!\!\!\! \!\!\!\!\!\!\!\!\!\!\!\! \!\!\!\!  &=& \!\!\!\!  \int_0^{\frac{y}{{K_s  - m}}} \!\! {\int_{\frac{y}{{K_s  - m}}}^{\frac{x}{m}} {p_{\sum\limits_{n = 1}^{m - 1} {\gamma _{n:K} } ,\gamma _{m:K} ,\sum\limits_{n = m + 1}^{K_s  - 1} {\gamma _{n:K} } ,\gamma _{K_s :K} } \!\! \left( {x - z_2 ,z_2 ,y - z_4 ,z_4 } \right)dz_2 } dz_4 }, \nonumber
\\
 \!\!\!\! \!\!\!\!\!\!\!\!\!\!\!\! \!\!\!\!  && \!\!\!\!\!\!\!\!\text{for } x>\frac{m}{K_s-m}y.
\end{eqnarray} \normalsize
\end{theorem}
\begin{proof}Omitted.
\end{proof}
Note again that only the finite integrations of joint PDFs are involved.

\section{Closed-form expressions for exponential RV case}
The above novel generic results are quite general and apply to any RVs. We now illustrate in this section some results for the i.i.d. exponential RV special case, where the PDF and the CDF of $\gamma$ are given by 
\begin{equation} \label{eq:pdf_of_Rayleigh}
\small p  \left( \gamma  \right)= \frac{1}{{\bar \gamma }}\exp \left( { - \frac{\gamma }{{\bar \gamma }}} \right),\quad\quad \text{for $\gamma\ge 0$},
\end{equation}
and
\begin{equation}
\small P  \left( \gamma  \right) = 1 - \exp \left( { - \frac{\gamma }{{\bar \gamma }}} \right),\quad\quad \text{for $\gamma\ge 0$},
\end{equation}
respectively, where $\bar{\gamma}$ is the common average. 
Therefore, (\ref{eq:CDF_MGF_multiple}), (\ref{eq:EDF_MGF_multiple}) and (\ref{eq:IntervalMGF_multiple}) specialize to

\small \begin{eqnarray} 
 \left[ {e\left( {z_a , - S_i } \right)} \right]^m  &=& \frac{{\left[ {\exp \left( { - \left( {S_i  + \frac{1}{{\bar \gamma }}} \right)z_a } \right)} \right]^m }}{{\left( {\bar \gamma } \right)^m \left( {S_i  + \frac{1}{{\bar \gamma }}} \right)^m }} \label{eq:common_function_Rayleigh_1}
\\ 
 \left[ {c\left( {z_a , - S_i } \right)} \right]^m  &=& \sum\limits_{j = 0}^m {\frac{{\left( { - 1} \right)^j }}{{\left( {\bar \gamma } \right)^m }}\binom{ m}{ j} \left[ {\exp \left( { - \frac{{z_a }}{{\bar \gamma }}} \right)} \right]^j \frac{{\left[ {\exp \left( { - z_a S_i } \right)} \right]^j }}{{\left( {S_i  + \frac{1}{{\bar \gamma }}} \right)^m }}} \label{eq:common_function_Rayleigh_2}
\\ 
 \left[ {\mu\left( {z_a ,z_b , - S_i } \right)} \right]^m  &=& \sum\limits_{j = 0}^m \Bigg[{\frac{{\left( { - 1} \right)^j }}{{\left( {\bar \gamma } \right)^m }}\binom{ m}{ j} \exp \left( { - \frac{{\left( {m - j} \right)}}{{\bar \gamma }}z_a } \right)} \exp \left( { - \frac{j}{{\bar \gamma }}z_b } \right) \nonumber
\\ 
  && \times \frac{{\exp \left( { - \left( {\left( {m - j} \right)z_a  + jz_b } \right)S_i } \right)}}{{\left( {S_i  + \frac{1}{{\bar \gamma }}} \right)^m }}\Bigg]. \label{eq:common_function_Rayleigh_3}
\end{eqnarray} \normalsize
After substituting (\ref{eq:common_function_Rayleigh_1}), (\ref{eq:common_function_Rayleigh_2}) and (\ref{eq:common_function_Rayleigh_3}) into the derived expressions of the joint PDF of partial sums of ordered statistics presented in the previous sections, it is easy to derive the following closed-form expressions for the PDFs by applying the classical inverse Laplace transform pair given in (\ref{eq:inverse_LT_pair}) and the property given in (\ref{eq:inverse_LT_property}). While some of these results have been derived using the successive conditioning approach previously, we list them here for the sake of convenience and completeness. 
\begin{enumerate}
\item[1)] PDF of $\sum\limits_{n = 1}^K {\gamma _{n:K} }$:
\begin{equation} \label{eq:closed_form_1}
\small p_Z \left( {z_1 } \right) = \frac{{z_1^{K - 1} }}{{\left( {K - 1} \right)!\bar \gamma ^K }}\exp \left( { - \frac{{z_1 }}{{\bar \gamma }}} \right).
\end{equation}
\item[2)] Joint PDF of $\gamma _{m:K}$ and $\sum\limits_{\scriptstyle n = 1 \hfill \atop \scriptstyle n \ne m \hfill}^K {\gamma _{n:K} } $:
\small \begin{equation} \label{eq:closed_form_2}
\!\!\!\!\!\!\!\!\!\!\!\! p_Z \left( {z_1 ,z_2 } \right) = 
\begin{cases}
\!\! \frac{{K!}}{{\left( {K - 1} \right)!\left( {K - 2} \right)!\bar \gamma ^K }}\exp \left( { - \frac{{z_1  + z_2 }}{{\bar \gamma }}} \right)
\\
\!\! \times \sum\limits_{j = 0}^{K - 1} {\left( { - 1} \right)^j \binom{K-1}{j} \left[ {z_2  - j z_1 } \right]^{K - 2} U\left( {z_2  - j z_1 } \right)} &\text{for }m=1,
\\
\!\! \frac{{K!}}{{\left( {K - m} \right)!\left( {m - 1} \right)!\left( {K - 2} \right)!\bar \gamma ^K }}\exp \left( { - \frac{{z_1  + z_2 }}{{\bar \gamma }}} \right)
\\ 
\!\! \times \sum\limits_{j = 0}^{K - m} {\left( { - 1} \right)^j \binom{K-m}{j} \left[ {z_2  - \left( {m + j - 1} \right)z_1 } \right]^{K - 2}} 
\\
\!\! \times U\left( {z_2  - \left( {m + j - 1} \right)z_1 } \right) &\text{for }m\ge2.
\end{cases}
\end{equation} \normalsize

\item[3)] Joint PDF of $\sum\limits_{n = 1}^m {\gamma _{n:K} }$ and $\sum\limits_{n = m + 1}^K {\gamma _{n:K} }$:
\begin{equation} \label{eq:closed_form_3}
\small p_Z \left( {z_1 ,z_2 } \right) = 
\begin{cases}
\frac{{K!}}{{\left( {K - m} \right)!\left( {K - m - 1} \right)!\left( {m - 1} \right)!\left( {m - 2} \right)!\bar \gamma ^K }}\exp \left( { - \frac{{z_1  + z_2 }}{{\bar \gamma }}} \right)
\\
 \times \int\limits_0^\infty  {d\gamma _{m:K} \Bigg[\left[ {z_1  - m\gamma _{m:K} } \right]^{m - 2} U\left( {z_1  - m\gamma _{m:K} } \right)} 
\\
 \times \sum\limits_{j = 0}^{K - m} {\left( { - 1} \right)^j \binom{K - m}{j}\left[ {z_2  - j\gamma _{m:K} } \right]^{K - m - 1} U\left( {z_2  - j\gamma _{m:K} } \right)}\Bigg], \quad &m\ge2
\\
\frac{{K!}}{{\left( {K - 1} \right)!\left( {K - 2} \right)!\bar \gamma ^K }}\exp \left( { - \frac{{z_1  + z_2 }}{{\bar \gamma }}} \right)
\\
 \times \sum\limits_{j = 0}^{K - 1} {\left( { - 1} \right)^j \binom{K-1}{j}\left[ {z_2  - jz_1 } \right]^{K - 2} U\left( {z_2  - jz_1 } \right)}, \quad &m=1.
\end{cases}
\end{equation}
\item[4)] PDF of $\sum\limits_{n = 1}^{K_s } {\gamma _{n:K} }$:
\begin{equation} \label{eq:closed_form_4}
\small p_{Z'} \left( {x } \right) =
\begin{cases}
\frac{{K!}}{{\left( {K - K_s } \right)!\left( {K_s  - 1} \right)!\left( {K_s  - 2} \right)!\bar \gamma ^{K_s } }}\exp \left( { - \frac{{x }}{{\bar \gamma }}} \right) \\
\times \int_0^{\frac{x}{K_s}}  { {\left[ {1 - \exp \left( { - \frac{{z_2 }}{{\bar \gamma }}} \right)} \right]^{K - K_s } \left[ {x  - K_s z_2 } \right]^{K_s  - 2} }}dz_2 ,\quad & K_s  \ge 2\\
\frac{K}{{\bar \gamma }}\exp \left( { - \frac{{x }}{{\bar \gamma }}} \right)\left[ {1 - \exp \left( { - \frac{{x }}{{\bar \gamma }}} \right)} \right]^{K - 1} ,\quad &K_s  = 1
\end{cases}
\end{equation}
where
\begin{equation}
\small \left[ {1 - \exp \left( { - \frac{{a }}{{\bar \gamma }}} \right)} \right]^{m}  = \sum\limits_{i = 0}^{m} {\left( { - 1} \right)^i \binom{m}{j}\left[ {\exp \left( { - \frac{{a }}{{\bar \gamma }}} \right)} \right]^i } .
\end{equation}

\item[5)] Joint PDF of $\gamma _{m:K}$ and $\sum\limits_{\scriptstyle n = 1 \hfill \atop \scriptstyle n \ne m \hfill}^{K_s } {\gamma _{n:K} }$:
\begin{enumerate}
\item[a.] For $m=1$,
\small \begin{eqnarray} \label{eq:closed_form_5}
\!\!\!\!\!\!\!\!\!\!\!\! p_Z \left( {z_1 ,z_2 ,z_3 } \right)\!\! &=& \!\! \frac{{K!}}{{\left( {K - K_s } \right)!\left( {K_s  - 2} \right)!\left( {K_s  - 3} \right)!\bar \gamma ^{K_s } }}\exp \left( { - \frac{{z_1  + z_2  + z_3 }}{{\bar \gamma }}} \right) \nonumber
\\ 
\!\!\!\!&&\!\! \times \left[ {1 - \exp \left( { - \frac{{z_3 }}{{\bar \gamma }}} \right)} \right]^{K - K_s } U\left( {z_1  - z_3 } \right) \nonumber
\\ 
\!\!\!\! && \!\! \times \sum\limits_{j = 0}^{K_s  - 2} \Bigg[\left( { - 1} \right)^j \binom{K_s-2}{j}\left[ {z_2  - \left( {K_s  - 2 - j} \right)z_3  - jz_1 } \right]^{K_s  - 3} \nonumber
\\
\!\!\!\! && \!\! \times U\left( {z_2  - \left( {K_s  - 2 - j} \right)z_3  - jz_1 } \right)\Bigg],
\end{eqnarray} \normalsize
\item[b.] For $1<m<K_s-1$,
\small \begin{eqnarray} \label{eq:closed_form_6}
\!\!\!\!\!\!\!\!\!\!\!\! &&\!\!\!\!\!\!\!\! p_Z \left( {z_1 ,z_2 ,z_3 ,z_4 } \right) \nonumber
\\
\!\!\!\!\!\!\!\!\!\!\!\! &=& \!\! \frac{{K!}}{{\left( {K - K_s } \right)!\left( {K_s  - m - 1} \right)!\left( {K_s  - m - 2} \right)!\left( {m - 1} \right)!\left( {m - 2} \right)!\bar \gamma ^{K_s } }} \nonumber
\\ 
\!\!\!\!\!\!\!\!\!\!\!\! &&\!\! \times \exp \left( { - \frac{{z_1  + z_2  + z_3  + z_4 }}{{\bar \gamma }}} \right)\left[ {1 - \exp \left( { - \frac{{z_4 }}{{\bar \gamma }}} \right)} \right]^{K - K_s } \left[ {z_1  - \left( {m - 1} \right)z_2 } \right]^{m - 2} \nonumber
\\ 
\!\!\!\!\!\!\!\!\!\!\!\! &&\!\! \times \sum\limits_{j = 0}^{K_s  - m - 1} \Bigg[{\left( { - 1} \right)^j \binom{K_s  - m - 1}{j}\left[ {z_3  - \left( {K_s  - m - 1 - j} \right)z_4  - jz_2 } \right]^{K_s  - m - 2} }  \nonumber
\\ 
\!\!\!\!\!\!\!\!\!\!\!\! &&\!\! \times U\left( {z_2  - z_4 } \right)U\left( {z_1  - \left( {m - 1} \right)z_2 } \right)U\left( {z_3  - \left( {K_s  - m - 1 - j} \right)z_4  - jz_2 } \right)\Bigg].
\end{eqnarray} \normalsize
\item[c.] For $m=K_s-1$,
\small \begin{eqnarray} \label{eq:closed_form_7}
\!\!\!\!\!\!\!\!\!\!\!\! p_Z \left( {z_1 ,z_2 ,z_3 } \right) \!\! &=& \!\! \frac{{K!}}{{\left( {K - K_s } \right)!\left( {K_s  - 2} \right)!\left( {K_s  - 3} \right)!\bar \gamma ^{K_s } }}\exp \left( { - \frac{{z_1  + z_2  + z_3 }}{{\bar \gamma }}} \right) \nonumber
\\
\!\!\!\! && \!\!\!\! \times \left[ {1 - \exp \left( { - \frac{{z_3 }}{{\bar \gamma }}} \right)} \right]^{K - K_s } U\left( {z_2  - z_3 } \right)\left[ {z_1  - \left( {K_s  - 2} \right)z_2 } \right]^{K_s  - 3} \nonumber
\\
\!\!\!\! && \!\!\!\! \times U\left( {z_1  - \left( {K_s  - 2} \right)z_2 } \right).
\end{eqnarray} \normalsize
\item[d.] For $m=K_s$,
\small \begin{eqnarray} \label{eq:closed_form_8}
\!\!\!\!\!\!\!\!\!\!\!\!\!\!\!\!\!\!\!\!\!\!\!\!\!\!\!\!\!\!    p_Z \left( {z_1 ,z_2 } \right) \!\! &=& \!\! \frac{{K!}}{{\left( {K - K_s } \right)!\left( {K_s  - 1} \right)!\left( {K_s  - 2} \right)!\bar \gamma ^{K_s } }}\exp \left( { - \frac{{z_1  + z_2 }}{{\bar \gamma }}} \right) \nonumber
\\ 
\!\!\!\!\!\!\!\!\!\!\!\!\!\!\!\!\!\!\!\!\!\!\!\!\!\!\!\!\!\!   && \!\!\!\! \times \left[ {1 - \exp \left( { - \frac{{z_1 }}{{\bar \gamma }}} \right)} \right]^{K - K_s } \left[ {z_2  - \left( {K_s  - 1} \right)z_1 }\right]^{K_s  - 2} U\left( {z_2  - \left( {K_s  - 1} \right)z_1 } \right).
\end{eqnarray} \normalsize
\end{enumerate}
\end{enumerate}


\section{Applications}

The above derived joint PDFs of partial sums of ordered statistics can be applied to the performance analysis of various wireless communication systems.
In this section, we discuss two examples. 

\subsection{Example 1)}
In conventional parallel multiuser scheduling schemes, a particular scheduled user's signal is detected by correlating signals of all scheduled users at the receiver. Therefore, under these practical conditions, every scheduled user is interfering with every other scheduled user. These effect is called Multiple Access Interference (MAI) or Multiple User Interference (MUI). This MUI is a factor which limits the capacity and performance of multiuser systems and the level of interference becomes substantial as the number of scheduled users increases. To take into account the effect of MUI caused by the other scheduled users, the signal to interference plus noise ratio (SINR) which measures the ratio between the useful power and the amount of noise and interference generated by all the other scheduled users is used. The level of interference becomes substantial as the number of the scheduled users increases because the SINR can decrease considerably in these conditions. This decrease in SINR can lead to a certain reduction of the rate allocated to each scheduled user.
With the above motivation in mind, the impact of interference on the performance (throughput) of the scheduled users assuming a selection based parallel multiuser scheduling scheme is needed to investigate the total average sum rate capacity and the average spectral efficiency (ASE) based on the SINR of the scheduled users. The major difficulty in investigating this total average sum rate capacity and ASE resides in the determination of the statistics of the SINR of the m-th scheduled user. Based on our approach, we can derive these results and then be applied to obtain the total average sum rate capacity and the ASE.
Based on this system model, the statistics of the SINR of the m-th scheduled user can be derived using the joint PDF of SNR of $m$-th desired scheduled user and the sum of the SNRs of the interfering $(K_S - 1)$ scheduled users among total $K$ $(K > K_S)$ users. Application to our proposed approach and the correspondent performance results over was presented in~\cite{kn:sungsiknam2008}.

\subsection{Example 2)}
Minimum selection generalized selection combining (MS-GSC) is an adaptive diversity combining scheme \cite{kn:Reed_MS, kn:Gupta_MS, kn:MS_GSC}. The basic idea is to combine a minimum number of best diversity paths among $L$ ones such that the combiner output SNR is above a certain preselected threshold, denoted by $\gamma_T$. The performance analysis of MS-GSC scheme is challenging due to the ordering operation on diversity paths and adaptive combining operation.  In~\cite{kn:Gupta_MS}, the average symbol error rate (SER) of MS-GSC was calculated as the weighted sum of the conditional average SER given that $m$ paths are combined, with the weights being the probabilities of combining $m$ paths and the average SER of conventional generalized selection combining (GSC), which always combines $m$ best paths, used as the conditional average SER. However, because the receiver with MS-GSC may combine $m$ best paths only under the condition that the combined SNR of first $m-1$  best paths is below the output threshold, the statistics of combined SNR with MS-GSC given that $m$ paths are combined is different from that with corresponding conventional GSC. As such, the average SER result in~\cite{kn:Gupta_MS} should serve as an approximation. To obtain the exact statistics of the combiner output SNR with MS-GSC, we need the joint statistics of the $m$th largest RV and the partial sum of the $m - 1$ largest RVs~\cite{kn:MS_GSC}.  Specifically, the exact CDF of the combiner output SNR with MS-GSC involves the probability $\Pr[\sum_{j=1}^{m-1}\gamma_{j:L} < \gamma_T ~\&~ 
\gamma_T \le \sum_{j=1}^{m}\gamma_{j:L} < x]$, which can be expressed in terms of the joint PDF of $\gamma_{m:L}$ and $\sum_{j=1}^{m-1}\gamma_{j:L}$.
With our proposed analytical framework, this joint PDF can be obtained in a systematic fashion, with the generic expression for generalized fading environments given in Eq. (36). This new statistical result allows us accurately evaluate the performance of  MS-GSC scheme over general fading channels.

Additionally, this proposed method can be applied to the performance analysis of various wireless communication systems with diversity techniques such that the performance analysis of GSC RAKE receiver with self-interference and so on. 

%

\section*{Appendices}

\appendices
\section{Useful inverse Laplace transform pair and property} \label{AP:Inverse_LT}
\setcounter{section}{1}

It is easy to see from (\ref{eq:closed_form_5}), (\ref{eq:closed_form_6}), (\ref{eq:closed_form_7}) and (\ref{eq:closed_form_8}) that the derivations of the PDF from the MGF involve the classical inverse Laplace transform pair \cite{kn:Mathematical_handbook}
\begin{equation} \label{eq:inverse_LT_pair}
\small {\mathcal{L}_s}^{-1}\left\{ \left( \frac{1}{s+a}\right)^n\right\} = \frac{1}{\left( n-1\right)!}t^{n-1}e^{-at},\quad t\ge0,n=1,2,3,\ldots,
\end{equation}
and the Laplace transform property \cite{kn:Mathematical_handbook}
\begin{equation} \label{eq:inverse_LT_property}
\small {\mathcal{L}_s}^{-1} \left\{ e^{-as}F\left(s\right)\right\} = f\left(t-a\right)U\left(t-a\right), \quad a>0.
\end{equation}

\section{Derivation of $I_m$} \label{AP:A}

In this appendix, we derive Eq. (\ref{eq:CDF_MGF_multiple}).

Let us first consider the case $m=K$. Noting that $p\left( {\gamma _{K:K} } \right)\exp \left( {\lambda \gamma _{K:K} } \right) = c'\left( {\gamma _{K:K} ,\lambda } \right)$,  we can rewrite $I_K$ as
\small \begin{eqnarray} \label{eq:appendix_A_2}
 \int\limits_0^{\gamma _{K - 1:K} } {d\gamma _{K:K} p\left( {\gamma _{K:K} } \right)\exp \left( {\lambda \gamma _{K:K} } \right)}  &=& \int\limits_0^{\gamma _{K - 1:K} } {d\gamma _{K:K} c'\left( {\gamma _{K:K} ,\lambda } \right)}
\\ \nonumber
  &=& \left. {c\left( {\gamma _{K:K} ,\lambda } \right)} \right|_0^{\gamma _{K - 1:K} }  
\\ \nonumber
  &=& c\left( {\gamma _{K - 1:K} ,\lambda } \right).
\end{eqnarray} \normalsize

For the case of $m=K-1$, after applying integration by part and  (\ref{eq:appendix_A_2}), we have
\small \begin{eqnarray} \label{eq:appendix_A_3}
&& \!\!\!\!\!\!\!\! \int\limits_0^{\gamma _{K - 2:K} } {d\gamma _{K - 1:K} p\left( {\gamma _{K - 1:K} } \right)\exp \left( {\lambda \gamma _{K - 1:K} } \right)\int\limits_0^{\gamma _{K - 1:K} } {d\gamma _{K:K} p\left( {\gamma _{K:K} } \right)\exp \left( {\lambda \gamma _{K:K} } \right)} }  \nonumber
\\
&=& \!\!\!\! \int\limits_0^{\gamma _{K - 2:K} } {d\gamma _{K - 1:K} c'\left( {\gamma _{K - 1:K} ,\lambda } \right)} c\left( {\gamma _{K - 1:K} ,\lambda } \right) \nonumber
\\ 
&=& \!\!\!\! \left. {\left[ {c\left( {\gamma _{K - 1:K} ,\lambda } \right)} \right]^2 } \right|_0^{\gamma _{K - 2:K} }  - \int\limits_0^{\gamma _{K - 2:K} } {d\gamma _{K - 1:K} c\left( {\gamma _{K - 1:K} ,\lambda } \right)c'\left( {\gamma _{K - 1:K} ,\lambda } \right)}.
\end{eqnarray} \normalsize
After moving the integration part of the right hand side (RHS) to the left hand side (LHS) and some manipulation, we can show
\begin{equation} \label{eq:appendix_A_4}
\small \int\limits_0^{\gamma _{K - 2:K} } {d\gamma _{K - 1:K} c'\left( {\gamma _{K - 1:K} ,\lambda } \right)c\left( {\gamma _{K - 1:K} ,\lambda } \right)}  = \frac{1}{2}\left[ {c\left( {\gamma _{K - 2:K} ,\lambda } \right)} \right]^2.
\end{equation}
Using (\ref{eq:appendix_A_4}) in (\ref{eq:appendix_A_3}), (\ref{eq:appendix_A_3}) can be re-written as
\begin{equation} \label{eq:appendix_A_5}
\small \int\limits_0^{\gamma _{K - 2:K} } {d\gamma _{K - 1:K} p\left( {\gamma _{K - 1:K} } \right)\exp \left( {\lambda \gamma _{K - 1:K} } \right)\int\limits_0^{\gamma _{K - 1:K} } {d\gamma _{K:K} p\left( {\gamma _{K:K} } \right)\exp \left( {\lambda \gamma _{K:K} } \right)} }  = \frac{1}{2}\left[ {c\left( {\gamma _{K - 2:K} ,\lambda } \right)} \right]^2 .
\end{equation}

Similarly for the case of $m=K-2$, with the help of integration by part, we can obtain the following
\small \begin{eqnarray} \label{eq:appendix_A_6}
\!\!\!\!\!\!\!\!\!\!&&\!\!\!\!\!\!\!\!\!\!\!\!\!\!\! \int\limits_0^{\gamma _{K - 3:K} }\!\!\!\! {d\gamma _{K - 2:K} p\left( {\gamma _{K - 2:K} } \right)\exp \left( {\lambda \gamma _{K - 2:K} } \right) \!\!\!\! \int\limits_0^{\gamma _{K - 2:K} } \!\!\!\! {d\gamma _{K - 1:K} p\left( {\gamma _{K - 1:K} } \right)\exp \left( {\lambda \gamma _{K - 1:K} } \right)} \!\!\!\! \int\limits_0^{\gamma _{K - 1:K} } \!\!\!\! {d\gamma _{K:K} p\left( {\gamma _{K:K} } \right)\exp \left( {\lambda \gamma _{K:K} } \right)} } \nonumber
\\
&&= \int\limits_0^{\gamma _{K - 3:K} } {d\gamma _{K - 2:K} c'\left( {\gamma _{K - 2:K} ,\lambda } \right)\frac{1}{2}\left[ {c\left( {\gamma _{K - 2:K} ,\lambda } \right)} \right]^2 } \nonumber
\\ 
&&= \frac{1}{2}\left. {\left[ {c\left( {\gamma _{K - 2:K} ,\lambda } \right)} \right]^3 } \right|_0^{\gamma _{K - 3:K} }  - \int\limits_0^{\gamma _{K - 3:K} } {d\gamma _{K - 2:K} \left[ {c\left( {\gamma _{K - 2:K} ,\lambda } \right)} \right]^2 c'\left( {\gamma _{K - 2:K} ,\lambda } \right)}.
\end{eqnarray} \normalsize
After some manipulation and substitution, we have \small \begin{eqnarray}\label{eq:appendix_A_8}
\!\!\!\!\!\!\!\!\!\!&&\!\!\!\!\!\!\!\!\!\!\!\!\!\!\! \int\limits_0^{\gamma _{K - 3:K} } \!\!\!\! {d\gamma _{K - 2:K} p\left( {\gamma _{K - 2:K} } \right)\exp \left( {\lambda \gamma _{K - 2:K} } \right) \!\!\!\! \int\limits_0^{\gamma _{K - 2:K} } \!\!\!\! {d\gamma _{K - 1:K} p\left( {\gamma _{K - 1:K} } \right)\exp \left( {\lambda \gamma _{K - 1:K} } \right)} \!\!\!\! \int\limits_0^{\gamma _{K - 1:K} } \!\!\!\! {d\gamma _{K:K} p\left( {\gamma _{K:K} } \right)\exp \left( {\lambda \gamma _{K:K} } \right)} } \nonumber
\\
&&= \frac{1}{{3 \times 2}}\left[ {c\left( {\gamma _{K - 3:K} ,\lambda } \right)} \right]^3.
\end{eqnarray} \normalsize

This process can be generalized to arbitrary $m$, which leads to the result in Eq. (\ref{eq:CDF_MGF_multiple}).

\section{Derivation of $I''_{a,b}$} \label{AP:C}

In this appendix, we show the derivation of Eq.(\ref{eq:IntervalMGF_multiple}).

Let
\small \begin{eqnarray} \label{eq:appendix_C_1}
 I''_{a,b} \!\!\!\! &=& \!\!\!\! \int\limits_{\gamma _{b:K} }^{\gamma _{a:K} } {d\gamma _{b - 1:K} \;p\left( {\gamma _{b - 1:K} } \right)\exp \left( {\lambda \gamma _{b - 1:K} } \right)\int\limits_{\gamma _{b - 1:K} }^{\gamma _{a:K} } {d\gamma _{b - 2:K} p\left( {\gamma _{b - 2:K} } \right)\exp \left( {\lambda \gamma _{b - 2:K} } \right)} }  \nonumber
\\ 
 \!\!\!\!  && \!\!\!\! \times \!\!\!\! \int\limits_{\gamma _{b - 2:K} }^{\gamma _{a:K} } {d\gamma _{b - 3:K} p\left( {\gamma _{b - 3:K} } \right)\exp \left( {\lambda \gamma _{b - 3:K} } \right)}  \cdots \!\!\!\! \int\limits_{\gamma _{a + 2:K} }^{\gamma _{a:K} } {d\gamma _{a + 1:K} p\left( {\gamma _{a + 1:K} } \right)\exp \left( {\lambda \gamma _{a + 1:K} } \right)}. 
\end{eqnarray} \normalsize

Using similar manipulations to the ones used in the previous Appendices \ref{AP:A}, we can write
\begin{equation}
\small \int\limits_{\gamma _{a + 2:K} }^{\gamma _{a:K} } {d\gamma _{a + 1:K} p\left( {\gamma _{a + 1:K} } \right)\exp \left( {\lambda \gamma _{a + 1:K} } \right)}  = \mu \left( {\gamma _{a + 2:K} ,\gamma _{a:K} ,\lambda } \right),
\end{equation}
\begin{equation}
\small \int\limits_{\gamma _{a + 3:K} }^{\gamma _{a:K} } {d\gamma _{a + 2:K} p\left( {\gamma _{a + 2:K} } \right)\exp \left( {\lambda \gamma _{a + 2:K} } \right)\int\limits_{\gamma _{a + 2:K} }^{\gamma _{a:K} } {d\gamma _{a + 1:K} p\left( {\gamma _{a + 1:K} } \right)\exp \left( {\lambda \gamma _{a + 1:K} } \right)} }  = \frac{1}{2}\left[ {\mu \left( {\gamma _{a + 3:K} ,\gamma _{a:K} ,\lambda } \right)} \right]^2,
\end{equation}
\small
\begin{eqnarray}
\!\!\!\!\!\!\!\!\!\!&&\!\!\!\!\!\!\!\!\!\!\!\!\!\!\! \int\limits_{\gamma _{a + 4:K} }^{\gamma _{a:K} } \!\!\!\!\! {d\gamma _{a + 3:K} p\left( {\gamma _{a + 3:K} } \right)\exp \left( {\lambda \gamma _{a + 3:K} } \right)\!\!\!\!\! \int\limits_{\gamma _{a + 3:K} }^{\gamma _{a:K} } \!\!\!\!\! {d\gamma _{a + 2:K} p\left( {\gamma _{a + 2:K} } \right)\exp \left( {\lambda \gamma _{a + 2:K} } \right)} \!\!\!\!\! \int\limits_{\gamma _{a + 2:K} }^{\gamma _{a:K} } \!\!\!\!\! {d\gamma _{a + 1:K} p\left( {\gamma _{a + 1:K} } \right)\exp \left( {\lambda \gamma _{a + 1:K} } \right)} } \nonumber
\\
&&= \frac{1}{{3 \times 2}}\left[ {\mu \left( {\gamma _{a + 4:K} ,\gamma _{a:K} ,\lambda } \right)} \right]^3.
\end{eqnarray}\normalsize

Using these results, $I''_{a,b}$ can be found in closed-form as
\small \begin{eqnarray}
 I''_{a,b} \!\!\!\! &=& \!\!\!\! \int\limits_{\gamma _{b:K} }^{\gamma _{a:K} } {d\gamma _{b - 1:K} \;p\left( {\gamma _{b - 1:K} } \right)\exp \left( {\lambda \gamma _{b - 1:K} } \right)\int\limits_{\gamma _{b - 1:K} }^{\gamma _{a:K} } {d\gamma _{b - 2:K} p\left( {\gamma _{b - 2:K} } \right)\exp \left( {\lambda \gamma _{b - 2:K} } \right)} } \nonumber 
\\ 
\!\!\!\! && \!\!\!\!\times \!\!\!\! \int\limits_{\gamma _{b - 2:K} }^{\gamma _{a:K} } {d\gamma _{b - 3:K} p\left( {\gamma _{b - 3:K} } \right)\exp \left( {\lambda \gamma _{b - 3:K} } \right)}  \cdots  \!\!\!\! \int\limits_{\gamma _{a + 2:K} }^{\gamma _{a:K} } {d\gamma _{a + 1:K} p\left( {\gamma _{a + 1:K} } \right)\exp \left( {\lambda \gamma _{a + 1:K} } \right)} \nonumber
\\ 
\!\!\!\! &=& \!\!\!\! \frac{1}{{\left( {b - a - 1} \right)!}}\left[ {\mu \left( {\gamma _{b:K} ,\gamma _{a:K} ,\lambda } \right)} \right]^{\left( {b - a - 1} \right)}.
\end{eqnarray} \normalsize

\section{Derivation of (\ref{eq:joint_MGF_1})}\label{AP:D}

Starting with (\ref{eq:joint_MGF_1}), by simply applying (\ref{eq:CDF_MGF_multiple}), we can obtain the following result easily
\begin{equation} \label{eq:joint_MGF_3_integralform} \small
\!\!\!\!\!\!\!\!\int\limits_0^{\gamma _{m:K} } \!\!\! {d\gamma _{m + 1:K} p\left( {\gamma _{m + 1:K} } \right)\exp \left( {\lambda _2 \gamma _{m + 1:K} } \right)}  \cdots \!\!\!\! \int\limits_0^{\gamma _{K - 1:K} } \!\!\! {d\gamma _{K:K} p\left( {\gamma _{K:K} } \right)\exp \left( {\lambda _2 \gamma _{K:K} } \right)} =\!\! \frac{1}{{\left( {K - m} \right)!}}\left[ {c\left( {\gamma _{m:K} ,\lambda _2 } \right)} \right]^{\left( {K - m} \right)}.
\end{equation} \normalsize
By inserting (\ref{eq:joint_MGF_3_integralform}) into (\ref{eq:joint_MGF_2_integralform}), the MGF has the following form:
\small
\begin{eqnarray} \label{eq:joint_MGF_4_integralform}
\!\!\!\! MGF_Z \left( {\lambda _1 ,\lambda _2 } \right) \!\!&=& \!\!  F\int\limits_0^\infty \!\! {d\gamma _{1:K} p\left( {\gamma _{1:K} } \right)\exp \left( {\lambda _2 \gamma _{1:K} } \right) \!\!\!\! \int\limits_0^{\gamma _{1:K} }\!\!\!\! {d\gamma _{2:K} p\left( {\gamma _{2:K} } \right)\exp \left( {\lambda _2 \gamma _{2:K} } \right)} } \nonumber
\\
&&\times \cdots \times \!\! \int\limits_0^{\gamma _{m - 2:K} } \!\!\!\! {d\gamma _{m - 1:K} p\left( {\gamma _{m - 1:K} } \right)\exp \left( {\lambda _2 \gamma _{m - 1:K} } \right)}   \nonumber
\\ 
 \!\!\!\!\!\!\!\!\!\!\!\! &&\times \!\! \int\limits_0^{\gamma _{m - 1:K} } \!\!\!\! {d\gamma _{m:K} p\left( {\gamma _{m:K} } \right)\exp \left( {\lambda _1 \gamma _{m:K} } \right)\frac{1}{{\left( {K - m} \right)!}}\left[ {c\left( {\gamma _{m:K} ,\lambda _2 } \right)} \right]^{\left( {K - m} \right)} }. 
\end{eqnarray} \normalsize
By applying the integral solution presented in (\ref{eq:Integral_solution}), we can re-write (\ref{eq:joint_MGF_4_integralform}) as the following:
\small
\begin{eqnarray} \label{eq:joint_MGF_5_integralform}
 \!\!\!\!\!\!\!\!\!\!\!\! MGF_Z \left( {\lambda _1 ,\lambda _2 } \right) \!\!&=& \!\! F\int\limits_0^\infty \!\! {d\gamma _{m:K} p\left( {\gamma _{m:K} } \right)\exp \left( {\lambda _1 \gamma _{m:K} } \right)\frac{1}{{\left( {K - m} \right)!}}\left[ {c\left( {\gamma _{m:K} ,\lambda _2 } \right)} \right]^{\left( {K - m} \right)} } \nonumber
\\ 
 \!\!\!\!\!\!\!\!\!\!\!\! &&\times \!\! \int\limits_{\gamma _{m:K} }^\infty \!\!\!\! {d\gamma _{m - 1:K} p\left( {\gamma _{m - 1:K} } \right)\exp \left( {\lambda _2 \gamma _{m - 1:K} } \right) \!\!\!\! \int\limits_{\gamma _{m - 1:K} }^\infty \!\!\!\! {d\gamma _{m - 2:K} p\left( {\gamma _{m - 2:K} } \right)\exp \left( {\lambda _2 \gamma _{m - 2:K} } \right)}} \nonumber
\\
&& \times \cdots \times \!\!\!\! \int\limits_{\gamma _{2:K} }^\infty \!\!\!\! {d\gamma _{1:K} p\left( {\gamma _{1:K} } \right)\exp \left( {\lambda _2 \gamma _{1:K} } \right)}.
\end{eqnarray} \normalsize
By simply applying (\ref{eq:EDF_MGF_multiple}), we can obtain the following result easily
\small \begin{eqnarray} \label{eq:joint_MGF_6_integralform}
 \!\!\!\!\!\!\!\!\!\!\!\! && \!\!\!\!\!\!\!\!\int\limits_{\gamma _{m:K} }^\infty \!\!\!\! {d\gamma _{m - 1:K} p\left( {\gamma _{m - 1:K} } \right)\exp \left( {\lambda _2 \gamma _{m - 1:K} } \right) \!\!\!\! \int\limits_{\gamma _{m - 1:K} }^\infty \!\!\!\! {d\gamma _{m - 2:K} p\left( {\gamma _{m - 2:K} } \right)\exp \left( {\lambda _2 \gamma _{m - 2:K} } \right)} } \nonumber 
\\
&& \times \cdots \times \!\!\!\! \int\limits_{\gamma _{2:K} }^\infty \!\!\!\! {d\gamma _{1:K} p\left( {\gamma _{1:K} } \right)\exp \left( {\lambda _2 \gamma _{1:K} } \right)}  \nonumber
\\
 \!\!\!\!\!\!\!\!\!\!\!\!&=& \!\! \frac{1}{{\left( {m - 1} \right)!}}\left[ {e\left( {\gamma _{m:K} ,\lambda _2 } \right)} \right]^{\left( {m - 1} \right)} 
\end{eqnarray} \normalsize
By inserting (\ref{eq:joint_MGF_6_integralform}) into (\ref{eq:joint_MGF_5_integralform}), we can obtain the second order MGF of $Z_1  = \gamma _{m:K}$ and $Z_2  = \sum\limits_{\scriptstyle n = 1 \hfill \atop \scriptstyle n \ne m \hfill}^K {\gamma _{n:K} }$ easily as the following:
\small\begin{equation} \label{eq:App_joint_MGF_1} 
 \!\!\!\!\!\!\!\!\!\! MGF_Z \left( {\lambda _1 ,\lambda _2 } \right) \!\!\!= \!\!\frac{F}{{\left( {K - m} \right)!\left( {m - 1} \right)!}}\int\limits_0^\infty  {d\gamma _{m:K} p\left( {\gamma _{m:K} } \right)\exp \left( {\lambda _1 \gamma _{m:K} } \right)\left[ {c\left( {\gamma _{m:K} ,\lambda _2 } \right)} \right]^{\left( {K - m} \right)} } \left[ {e\left( {\gamma _{m:K} ,\lambda _2 } \right)} \right]^{\left( {m - 1} \right)}.
\end{equation}\normalsize

\section{Derivation of (\ref{eq:joint_MGF_2})}\label{AP:E}

Starting with (\ref{eq:joint_MGF_2}), by simply applying (\ref{eq:CDF_MGF_multiple}), we can obtain the following result easily
\small \begin{eqnarray} \label{eq:joint_MGF_8_integralform}
\!\!\!\!\!\!\!\!&&\!\!\!\!\!\!\!\! \int\limits_0^{\gamma _{m:K} } {d\gamma _{m + 1:K} p\left( {\gamma _{m + 1:K} } \right)\exp \left( {\lambda _2 \gamma _{m + 1:K} } \right)}  \cdots \!\!\!\! \int\limits_0^{\gamma _{K - 1:K} } {d\gamma _{K:K} p\left( {\gamma _{K:K} } \right)\exp \left( {\lambda _2 \gamma _{K:K} } \right)} \nonumber
\\
&=&\!\! \frac{1}{{\left( {K - m} \right)!}}\left[ {c\left( {\gamma _{m:K} ,\lambda _2 } \right)} \right]^{\left( {K - m} \right)}.
\end{eqnarray} \normalsize
By inserting (\ref{eq:joint_MGF_8_integralform}) into (\ref{eq:joint_MGF_7_integralform}), the MGF has the following form:
\small \begin{eqnarray} \label{eq:joint_MGF_9_integralform}
\!\!\!\!\!\!\!\!\!\!\!\!\!\! MGF_Z \left( {\lambda _1 ,\lambda _2 } \right) \!\!&=& \!\! F  \!\! \int\limits_0^\infty  \!\! {d\gamma _{1:K} p\left( {\gamma _{1:K} } \right)\exp \left( {\lambda _1 \gamma _{1:K} } \right) }\nonumber
\\
&& \times \cdots \times \!\!\!\! \!\!\!\int\limits_0^{\gamma _{m - 1:K} }  \!\!\!\! {d\gamma _{m:K} p\left( {\gamma _{m:K} } \right)\exp \left( {\lambda _1 \gamma _{m:K} } \right)\frac{1}{{\left( {K - m} \right)!}}\left[ {c\left( {\gamma _{m:K} ,\lambda _2 } \right)} \right]^{\left( {K - m} \right)} }.
\end{eqnarray} \normalsize
By applying the integral solution presented in (\ref{eq:Integral_solution}), we can re-write (\ref{eq:joint_MGF_9_integralform}) as the following:
\small \begin{eqnarray} \label{eq:joint_MGF_10_integralform}
\!\!\!\!\!\!\!\!\!\!\!\! MGF_Z \left( {\lambda _1 ,\lambda _2 } \right) \!\!&=& \!\! F\int\limits_0^\infty  {d\gamma _{m:K} p\left( {\gamma _{m:K} } \right)\exp \left( {\lambda _1 \gamma _{m:K} } \right)\frac{1}{{\left( {K - m} \right)!}}\left[ {c\left( {\gamma _{m:K} ,\lambda _2 } \right)} \right]^{\left( {K - m} \right)} }  \nonumber
\\ 
 \!\!\!\!\!\!\!\!\!\!\!\! &&\times \!\! \int\limits_{\gamma _{m:K} }^\infty  {d\gamma _{m - 1:K} p\left( {\gamma _{m - 1:K} } \right)\exp \left( {\lambda _1 \gamma _{m - 1:K} } \right) \cdots \!\!\!\! \int\limits_{\gamma _{2:K} }^\infty  {d\gamma _{1:K} p\left( {\gamma _{1:K} } \right)\exp \left( {\lambda _1 \gamma _{1:K} } \right)} }. 
\end{eqnarray} \normalsize
By simply applying (\ref{eq:EDF_MGF_multiple}), we can obtain the following result easily
\small \begin{eqnarray} \label{eq:joint_MGF_11_integralform}
&&\!\!\!\!\!\! \int\limits_{\gamma _{m:K} }^\infty  {d\gamma _{m - 1:K} p\left( {\gamma _{m - 1:K} } \right)\exp \left( {\lambda _1 \gamma _{m - 1:K} } \right) \cdots \!\!\!\! \int\limits_{\gamma _{2:K} }^\infty  {d\gamma _{1:K} p\left( {\gamma _{1:K} } \right)\exp \left( {\lambda _1 \gamma _{1:K} } \right)} } \nonumber
\\
&=& \frac{1}{{\left( {m - 1} \right)!}}\left[ {e\left( {\gamma _{m:K} ,\lambda _1 } \right)} \right]^{\left( {m - 1} \right)} .
\end{eqnarray} \normalsize
Substituting (\ref{eq:joint_MGF_11_integralform}) in (\ref{eq:joint_MGF_10_integralform}), we can obtain the second order MGF of $Z_1  = \sum\limits_{n = 1}^m {\gamma _{n:K} }$ and $Z_2  = \sum\limits_{n = m + 1}^K {\gamma _{n:K} }$ as
\small
\begin{eqnarray} \label{eq:APP_joint_MGF_2}
\!\!\!\!\!\!\!\!\!\!\!\!\!\!\!\!\!\!\!\!\! && \!\!\!\! MGF_Z \left( {\lambda _1 ,\lambda _2 } \right) \nonumber
\\
\!\!\!\!\!\!\!\!\!\!\!\!\!\!\!\!\!\!\!\!\! &=& \!\! \frac{{K!}}{{\left( {K - m} \right)!\left( {m - 1} \right)!}}\int\limits_0^\infty  {d\gamma _{m:K} p\left( {\gamma _{m:K} } \right)\exp \left( {\lambda _1 \gamma _{m:K} } \right)\left[ {c\left( {\gamma _{m:K} ,\lambda _2 } \right)} \right]^{\left( {K - m} \right)} \left[ {e\left( {\gamma _{m:K} ,\lambda _1 } \right)} \right]^{\left( {m - 1} \right)} }.
\end{eqnarray} \normalsize

\section{Derivation of the joint PDF of $\gamma _{m:K}$ and $\sum\limits_{\scriptstyle n = 1 \hfill \atop \scriptstyle n \ne m \hfill}^{K_s } {\gamma _{n:K} }$ among $K$ ordered RVs}\label{AP:G}

In this Appendix, we derive the joint PDF of $\gamma _{m:K}$ and $\sum\limits_{\scriptstyle n = 1 \hfill \atop \scriptstyle n \ne m \hfill}^{K_s } {\gamma _{n:K} }$ among $K$ ordered RVs by considering four cases i) $m=1$, ii) $1<m<K_s-1$, iii) $m=K_s-1$ and iv) $m=K_s$ separately.

Consider first the case ii), $1<m<K_s-1$. Let $Z_1  = \sum\limits_{n = 1}^{m - 1} {\gamma _{n:K} }$, $Z_2  = \gamma _{m:K}$, $Z_3 = \sum\limits_{n = m + 1}^{K_s  - 1} {\gamma _{n:K} }$ and $Z_4  = \gamma _{K_s :K}$. The 4-dimensional MGF of $Z=\left[Z_1,Z_2,Z_3,Z_4\right]$ is given by the expectation
\small \begin{eqnarray} \label{eq:APP_joint_MGF_GSC_2_integralform}
\!\!\!\!\!\!\!\!\!\!\!\! && \!\!\!\!\!\!\!\! MGF_Z \left( {\lambda _1 ,\lambda _2 ,\lambda _3 ,\lambda _4 } \right) \!\! = \!\! E\left\{ {\exp \left( {\lambda _1 Z_1  + \lambda _2 Z_2  + \lambda _3 Z_3  + \lambda _4 Z_4 } \right)} \right\} \nonumber
\\ 
\!\!\!\!\!\!\!\!\!\!\!\!  &=& \!\!\!\! F\int\limits_0^\infty  {d\gamma _{1:K} p\left( {\gamma _{1:K} } \right)\exp \left( {\lambda _1 \gamma _{1:K} } \right) \cdots \!\!\!\! \int\limits_0^{\gamma _{m - 2:K} } {d\gamma _{m - 1:K} p\left( {\gamma _{m - 1:K} } \right)\exp \left( {\lambda _1 \gamma _{m - 1:K} } \right)} }  \nonumber
\\ 
\!\!\!\!\!\!\!\!\!\!\!\!  && \!\!\!\! \times \!\!\!\!\! \int\limits_0^{\gamma _{m - 1:K} } {d\gamma _{m:K} p\left( {\gamma _{m:K} } \right)\exp \left( {\lambda _2 \gamma _{m:K} } \right)} \nonumber
\\ 
\!\!\!\!\!\!\!\!\!\!\!\!  && \!\!\!\! \times \!\! \int\limits_0^{\gamma_{m:K}}  {d\gamma _{m + 1:K} p\left( {\gamma _{m + 1:K} } \right)\exp \left( {\lambda _3 \gamma _{m + 1:K} } \right) \cdots \!\!\!\! \int\limits_0^{\gamma _{K_s  - 2:K} } {d\gamma _{K_s  - 1:K} p\left( {\gamma _{K_s  - 1:K} } \right)\exp \left( {\lambda _3 \gamma _{K_s  - 1:K} } \right)} } \nonumber
\\ 
\!\!\!\!\!\!\!\!\!\!\!\!  && \!\!\!\! \times \!\!\!\!\! \int\limits_0^{\gamma _{K_s  - 1:K} } {d\gamma _{K_s :K} p\left( {\gamma _{K_s :K} } \right)\exp \left( {\lambda _4 \gamma _{K_s :K} } \right)\left[ {c\left( {\gamma _{K_s :K} } \right)} \right]^{\left( {K - K_s } \right)} }. 
\end{eqnarray} \normalsize
With the help of (\ref{eq:Integral_solution}), (\ref{eq:CDF_MGF_multiple}), (\ref{eq:EDF_MGF_multiple}) and (\ref{eq:IntervalMGF_multiple}), we can easily obtain the 4-dimensional MGF of $Z_1  = \sum\limits_{n = 1}^{m - 1} {\gamma _{n:K} }$, $Z_2  = \gamma _{m:K}$, $Z_3 = \sum\limits_{n = m + 1}^{K_s  - 1} {\gamma _{n:K} }$ and $Z_4  = \gamma _{K_s :K}$ as
\small \begin{eqnarray} \label{eq:APP_joint_MGF_GSC_2}
&&\!\!\!\!\!\!\!\! MGF_Z \left( {\lambda _1 ,\lambda _2 ,\lambda _3 ,\lambda _4 } \right) \nonumber
\\
&=& \!\!\!\! \frac{F}{{\left( {K_s  - m - 1} \right)!\left( {m - 1} \right)!}}\int\limits_0^\infty  {d\gamma _{K_s :K} p\left( {\gamma _{K_s :K} } \right)\exp \left( {\lambda _4 \gamma _{K_s :K} } \right)\left[ {c\left( {\gamma _{K_s :K} } \right)} \right]^{\left( {K - K_s } \right)} } \nonumber
\\
&& \!\!\!\! \times \!\!\!\! \int\limits_{\gamma _{K_s :K} }^\infty  {d\gamma _{m:K} p\left( {\gamma _{m:K} } \right)\exp \left( {\lambda _2 \gamma _{m:K} } \right)\left[ {e\left( {\gamma _{m:K} ,\lambda _1 } \right)} \right]^{\left( {m - 1} \right)} \left[ {\mu \left( {\gamma _{K_s :K} ,\gamma _{m:K} ,\lambda _3 } \right)} \right]^{\left( {K_s  - m - 1} \right)} }.
\end{eqnarray} \normalsize

Having a MGF expression given in (\ref{eq:APP_joint_MGF_GSC_2}), we are now in the position to derive the 4-dimensional joint PDF of $Z_1  = \sum\limits_{n = 1}^{m - 1} {\gamma _{n:K} }$, $Z_2  = \gamma _{m:K}$, $Z_3 = \sum\limits_{n = m + 1}^{K_s  - 1} {\gamma _{n:K} }$ and $Z_4  = \gamma _{K_s :K}$. Letting $\lambda _1  =  - S_1$, $\lambda _2  =  - S_2$, $\lambda _3  =  - S_3$,  and $\lambda _4  =  - S_4$  we can derive the 4-dimensional PDF of $Z_1  = \sum\limits_{n = 1}^{m - 1} {\gamma _{n:K} }$, $Z_2  = \gamma _{m:K}$, $Z_3 = \sum\limits_{n = m + 1}^{K_s  - 1} {\gamma _{n:K} }$ and $Z_4  = \gamma _{K_s :K}$ by applying an inverse Laplace transform yielding
\small \begin{eqnarray} \label{eq:APP_joint_PDF_GSC_2}
\!\!\!\!\!\!\!\!\!\!\!\!\!\!\!\! && \!\!\!\!\!\!\!\! p_Z \left( {z_1 ,z_2 ,z_3 ,z_4 } \right) = \mathcal{L}_{S_1 ,S_2 ,S_3 ,S_4 }^{ - 1} \left\{ {MGF_Z \left( { - S_1 , - S_2 , - S_3 , - S_4 } \right)} \right\} \nonumber
\\ 
\!\!\!\!\!\!\!\!\!\!\!\!\!\!\!\!  &=& \!\!\!\! \frac{F}{{\left( {K_s  - m - 1} \right)!\left( {m - 1} \right)!}}\int\limits_0^\infty  {d\gamma _{K_s :K} p\left( {\gamma _{K_s :K} } \right)\mathcal{L}_{S_4 }^{ - 1} \left\{ {\exp \left( { - S_4 \gamma _{K_s :K} } \right)} \right\})\left[ {c\left( {\gamma _{K_s :K} } \right)} \right]^{\left( {K - K_s } \right)} }  \nonumber
\\ 
\!\!\!\!\!\!\!\!\!\!\!\!\!\!\!\!  && \!\!\!\! \times \!\!\!\! \int\limits_{\gamma _{K_s :K} }^\infty \!\!\!\! {d\gamma _{m:K} \Bigg[p\left( {\gamma _{m:K} } \right)\mathcal{L}_{S_2 }^{ - 1} \left\{ {\exp \left( { - S_2 \gamma _{m:K} } \right)} \right\}\mathcal{L}_{S_1 }^{ - 1} \left\{ {\left[ {e\left( {\gamma _{m:K} , - S_1 } \right)} \right]^{\left( {m - 1} \right)} } \right\}} \nonumber
\\
\!\!\!\!\!\!\!\!\!\!\!\!\!\!\!\!  && \!\!\!\! \times \mathcal{L}_{S_3 }^{ - 1} \left\{ {\left[ {\mu \left( {\gamma _{K_s :K} ,\gamma _{m:K} , - S_3 } \right)} \right]^{\left( {K_s  - m - 1} \right)} } \right\}\Bigg]  \nonumber
\\ 
\!\!\!\!\!\!\!\!\!\!\!\!\!\!\!\!  &=& \!\!\!\! \frac{F}{{\left( {K_s  - m - 1} \right)!\left( {m - 1} \right)!}}p\left( {z_2 } \right)p\left( {z_4 } \right)\left[ {c\left( {z_4 } \right)} \right]^{\left( {K - K_s } \right)} U\left( {z_2  - z_4 } \right) \nonumber
\\ 
\!\!\!\!\!\!\!\!\!\!\!\!\!\!\!\!  && \!\!\!\! \times \mathcal{L}_{S_1 }^{ - 1} \left\{ {\left[ {e\left( {z_2 , - S_1 } \right)} \right]^{\left( {m - 1} \right)} } \right\}\mathcal{L}_{S_3 }^{ - 1} \left\{ {\left[ {\mu \left( {z_4 ,z_2 , - S_3 } \right)} \right]^{\left( {K_s  - m - 1} \right)} } \right\}. 
\end{eqnarray} \normalsize

With this 4-dimensional joint PDF, letting $X=Z_2$ and $Y=Z_1+Z_3+Z_4$ we can obtain the 2-dimensional joint PDF of $Z^{'}=[X,Y]$ by integrating over $z_1$ and $z_4$ yielding
\begin{equation} \label{eq:AP_final_1}
\small p_{Z^{'}} \left( {x,y} \right) = \int_0^x {\int_{\left( {m - 1} \right)x}^{y - \left(K_s-m\right)z_4 } {p_Z \left( {z_1 ,x,y - z_4 ,z_4 } \right)dz_1 } dz_4 }, 
\end{equation}
or equivalently we can obtain the 2-dimensional joint PDF of $Z^{'}=[X,Y]$ by integrating over $z_3$ and $z_4$ giving
\begin{equation} \label{eq:AP_final_2}
\small p_{Z^{'}} \left( {x,y} \right) = \int_0^x {\int_{\left( {K_s  - m - 1} \right)z_4 }^{\left( {K_s  - m - 1} \right)x} {p_Z \left( {y - z_3  - z_4 ,x,z_3 ,z_4 } \right)dz_3 } dz_4 } .
\end{equation}

We now consider the case i) for which $m=1$. Let $Z_1  = \gamma _{1:K}$, $Z_2  = \sum\limits_{n = 2}^{K_s  - 1} {\gamma _{n:K} }$ and $Z_3  = \gamma _{K_s :K}
$ for convenience. For this case, the 3-dimensional MGF of $Z=\left[Z_1,Z_2,Z_3\right]$ is given by the expectation
\small \begin{eqnarray} \label{eq:APP_joint_MGF_GSC_3_integralform}
\!\!\!\!\!\!\!\!\!\!\!\! && \!\!\!\!\!\!\!\! MGF_Z \left( {\lambda _1 ,\lambda _2 ,\lambda _3 } \right) = E\left\{ {\exp \left( {\lambda _1 Z_1  + \lambda _2 Z_2  + \lambda _3 Z_3 } \right)} \right\} \nonumber
\\ 
\!\!\!\!\!\!\!\!\!\!\!\! &=& \!\!\!\! F\int\limits_0^\infty  {d\gamma _{1:K} p\left( {\gamma _{1:K} } \right)\exp \left( {\lambda _1 \gamma _{1:K} } \right)}  \nonumber
\\ 
\!\!\!\!\!\!\!\!\!\!\!\! && \!\!\!\!  \times \!\!\!\! \int\limits_0^{\gamma _{1:K} } {d\gamma _{2:K} p\left( {\gamma _{2:K} } \right)\exp \left( {\lambda _2 \gamma _{2:K} } \right) \cdots \!\!\!\! \int\limits_0^{\gamma _{K_s  - 2:K} } {d\gamma _{K_s  - 1:K} p\left( {\gamma _{K_s  - 1:K} } \right)\exp \left( {\lambda _2 \gamma _{K_s  - 1:K} } \right)} }  \nonumber
\\ 
\!\!\!\!\!\!\!\!\!\!\!\! && \!\!\!\!  \times \!\!\!\! \int\limits_0^{\gamma _{K_s  - 1:K} } {d\gamma _{K_s :K} p\left( {\gamma _{K_s :K} } \right)\exp \left( {\lambda _3 \gamma _{K_s :K} } \right)\left[ {c\left( {\gamma _{K_s :K} } \right)} \right]^{\left( {K - K_s } \right)} }.
\end{eqnarray} \normalsize
With the help of (\ref{eq:Integral_solution}) and (\ref{eq:IntervalMGF_multiple}), we can easily obtain the 3-dimensional MGF of $Z_1  = \gamma _{1:K}$, $Z_2  = \sum\limits_{n = 2}^{K_s  - 1} {\gamma _{n:K} }$ and $Z_3  = \gamma _{K_s :K}$ as
\small \begin{eqnarray} \label{eq:APP_joint_MGF_GSC_3}
\!\!\!\!\!\!\!\!\!\!\!\! MGF_Z \left( {\lambda _1 ,\lambda _2 ,\lambda _3 } \right)\!\! &=& \!\!\!\! F\int\limits_0^\infty  {d\gamma _{K_s :K} p\left( {\gamma _{K_s :K} } \right)\exp \left( {\lambda _3 \gamma _{K_s :K} } \right)\left[ {c\left( {\gamma _{K_s :K} } \right)} \right]^{\left( {K - K_s } \right)} }  \nonumber
\\ 
\!\!\!\!\!\!\!\!\!\!\!\! && \!\!\!\!  \times \!\!\!\! \int\limits_{\gamma _{K_s :K} }^\infty  {d\gamma _{1:K} p\left( {\gamma _{1:K} } \right)\exp \left( {\lambda _1 \gamma _{1:K} } \right)\frac{1}{{\left( {K_s  - 2} \right)!}}\left[ {\mu \left( {\gamma _{K_s :K} ,\gamma _{1:K} ,\lambda _2 } \right)} \right]^{\left( {K_s  - 2} \right)} } . 
\end{eqnarray} \normalsize

Having an expression from the MGF as given in (\ref{eq:APP_joint_MGF_GSC_3}), we are now in the position to derive the 3-dimensional joint PDF of $Z_1  = \gamma _{1:K}$, $Z_2  = \sum\limits_{n = 2}^{K_s  - 1} {\gamma _{n:K} }$ and $Z_3  = \gamma _{K_s :K}$. Letting $\lambda _1  =  - S_1$, $\lambda _2  =  - S_2$ and $\lambda _3  =  - S_3$,  we can derive the 3-dimensional joint PDF of $Z_1  = \gamma _{1:K}$, $Z_2  = \sum\limits_{n = 2}^{K_s  - 1} {\gamma _{n:K} }$ and $Z_3  = \gamma _{K_s :K}$ by applying an inverse Laplace transform yielding
\small \begin{eqnarray}  \label{eq:APP_joint_PDF_GSC_3}
\!\!\!\!\!\!\!\!\!\!\!\! && \!\!\!\!\!\!\!\! p_Z \left( {z_1 ,z_2 ,z_3 } \right) = \mathcal{L}_{S_1 ,S_2 ,S_3 }^{ - 1} \left\{ {MGF_Z \left( { - S_1 , - S_2 , - S_3 } \right)} \right\} \nonumber
\\ 
\!\!\!\!\!\!\!\!\!\!\!\! &=& \!\!\!\! F\int\limits_0^\infty  {d\gamma _{K_s :K} p\left( {\gamma _{K_s :K} } \right)\mathcal{L}_{S_3 }^{ - 1} \left\{ {\exp \left( {-S_3 \gamma _{K_s :K} } \right)} \right\}\left[ {c\left( {\gamma _{K_s :K} } \right)} \right]^{\left( {K - K_s } \right)} }  \nonumber
\\ 
\!\!\!\!\!\!\!\!\!\!\!\! && \!\!\!\!  \times \!\!\!\! \int\limits_{\gamma _{K_s :K} }^\infty  {d\gamma _{1:K} p\left( {\gamma _{1:K} } \right)\mathcal{L}_{S_1 }^{ - 1} \left\{ {\exp \left( { - S_1 \gamma _{1:K} } \right)} \right\}\frac{1}{{\left( {K_s  - 2} \right)!}}\mathcal{L}_{S_2 }^{ - 1} \left\{ {\left[ {\mu \left( {\gamma _{K_s :K} ,\gamma _{1:K} , - S_2 } \right)} \right]^{\left( {K_s  - 2} \right)} } \right\}}  \nonumber
\\ 
\!\!\!\!\!\!\!\!\!\!\!\! &=& \!\!\!\! \frac{F}{{\left( {K_s  - 2} \right)!}}\int\limits_0^\infty  {d\gamma _{K_s :K} p\left( {\gamma _{K_s :K} } \right)\delta \left( {z_3  - \gamma _{K_s :K} } \right)\left[ {c\left( {\gamma _{K_s :K} } \right)} \right]^{\left( {K - K_s } \right)} }  \nonumber
\\ 
\!\!\!\!\!\!\!\!\!\!\!\! && \!\!\!\!  \times \!\!\!\! \int\limits_{\gamma _{K_s :K} }^\infty  {d\gamma _{1:K} p\left( {\gamma _{1:K} } \right)\delta \left( {z_1  - \gamma _{1:K} } \right)\frac{1}{{\left( {K_s  - 2} \right)!}}\mathcal{L}_{S_2 }^{ - 1} \left\{ {\left[ {\mu \left( {\gamma _{K_s :K} ,\gamma _{1:K} , - S_2 } \right)} \right]^{\left( {K_s  - 2} \right)} } \right\}}  \nonumber
\\ 
\!\!\!\!\!\!\!\!\!\!\!\! &=& \!\!\!\! \frac{F}{{\left( {K_s  - 2} \right)!}}p\left( {z_1 } \right)p\left( {z_3 } \right)\left[ {c\left( {z_3 } \right)} \right]^{\left( {K - K_s } \right)} U\left( {z_1  - z_3 } \right)\mathcal{L}_{S_2 }^{ - 1} \left\{ {\left[ {\mu \left( {z_3 ,z_1 , - S_2 } \right)} \right]^{\left( {K_s  - 2} \right)} } \right\} .
\end{eqnarray} \normalsize

With these 3-dimensional joint PDF, letting $X=Z_1$ and $Y=Z_2+Z_3$, we can obtain the 2-dimensional joint PDF of $Z^{'}=[X,Y]$ by integrating over $z_2$ yielding
\begin{equation} \label{eq:AP_final_3}
\small p_{Z^{'}} \left( {x,y} \right) = \int_{\left( {\frac{{K_s  - 2}}{{K_s  - 1}}} \right)y}^{\left( {K_s  - 2} \right)x} {p_Z \left( {x,z_2 ,y - z_2 } \right)dz_2 }.
\end{equation}

We now consider the case iii) for which $m=K_s-1$. Let $Z_1  = \sum\limits_{n = 1}^{K_s  - 2} {\gamma _{n:K} }$, $Z_2  = \gamma _{K_s  - 1:K}$ and $Z_3  = \gamma _{K_s :K}$. The 3-dimensional MGF of $Z=\left[Z_1,Z_2,Z_3\right]$ is given by
\small \begin{eqnarray} \label{eq:APP_joint_MGF_GSC_4_integralform}
\!\!\!\!\!\!\!\!\!\!\!\! && \!\!\!\!\!\!\!\! MGF_Z \left( {\lambda _1 ,\lambda _2 ,\lambda _3 } \right) = E\left\{ {\exp \left( {\lambda _1 Z_1  + \lambda _2 Z_2  + \lambda _3 Z_3 } \right)} \right\} \nonumber
\\ 
\!\!\!\!\!\!\!\!\!\!\!\! &=& \!\!\!\!  F\int\limits_0^\infty  {d\gamma _{1:K} p\left( {\gamma _{1:K} } \right)\exp \left( {\lambda _1 \gamma _{1:K} } \right) \cdots \!\!\!\! \int\limits_0^{\gamma _{K_s  - 3:K} } {d\gamma _{K_s  - 2:K} p\left( {\gamma _{K_s  - 2:K} } \right)\exp \left( {\lambda _1 \gamma _{K_s  - 2:K} } \right)} }  \nonumber
\\ 
\!\!\!\!\!\!\!\!\!\!\!\! && \!\!\!\!  \times \!\!\!\!\! \int\limits_0^{\gamma _{K_s  - 2:K} } \!\! {d\gamma _{K_s  - 1:K} p\left( {\gamma _{K_s  - 1:K} } \right)\exp \left( {\lambda _2 \gamma _{K_s  - 1:K} } \right)} \!\!\!\!\!\! \int\limits_0^{\gamma _{K_s  - 1:K} } \!\! {d\gamma _{K_s :K} p\left( {\gamma _{K_s :K} } \right)\exp \left( {\lambda _3 \gamma _{K_s :K} } \right)\left[ {c\left( {\gamma _{K_s :K} } \right)} \right]^{\left( {K - K_s } \right)} }.\ 
\end{eqnarray} \normalsize
With the help of (\ref{eq:Integral_solution}) and (\ref{eq:EDF_MGF_multiple}), we can easily obtain the 3-dimensional MGF of $Z_1  = \sum\limits_{n = 1}^{K_s  - 2} {\gamma _{n:K} }$, $Z_2  = \gamma _{K_s  - 1:K}$ and $Z_3  = \gamma _{K_s :K}$ as
\small \begin{eqnarray} \label{eq:APP_joint_MGF_GSC_4}
\!\!\!\!\!\!\!\!\!\!\!\! MGF_Z \left( {\lambda _1 ,\lambda _2 ,\lambda _3 } \right) \!\! &=& \!\!\!\! F\int\limits_0^\infty  {d\gamma _{K_s :K} p\left( {\gamma _{K_s :K} } \right)\exp \left( {\lambda _3 \gamma _{K_s :K} } \right)\left[ {c\left( {\gamma _{K_s :K} } \right)} \right]^{\left( {K - K_s } \right)} }  \nonumber
\\ 
\!\!\!\!\!\!\!\!\!\!\!\! && \!\!\!\!  \times \!\!\!\! \int\limits_{\gamma _{K_s :K} }^\infty  {d\gamma _{K_s  - 1:K} p\left( {\gamma _{K_s  - 1:K} } \right)\exp \left( {\lambda _2 \gamma _{K_s  - 1:K} } \right)\frac{1}{{\left( {K_s  - 2} \right)!}}\left[ {e\left( {\gamma _{K_s  - 1:K} ,\lambda _1 } \right)} \right]^{\left( {K_s  - 2} \right)} }.
\end{eqnarray} \normalsize

Having a MGF expression as given in (\ref{eq:APP_joint_MGF_GSC_4}), we are now in the position to derive the 3-dimensional joint PDF of $Z_1  = \sum\limits_{n = 1}^{K_s  - 2} {\gamma _{n:K} }$, $Z_2  = \gamma _{K_s  - 1:K}$ and $Z_3  = \gamma _{K_s :K}$. Letting $\lambda _1  =  - S_1$, $\lambda _2  =  - S_2$ and $\lambda _3  =  - S_3$  we can derive the 3-dimensional joint PDF of $Z_1  = \sum\limits_{n = 1}^{K_s  - 2} {\gamma _{n:K} }$, $Z_2  = \gamma _{K_s  - 1:K}$ and $Z_3  = \gamma _{K_s :K}$ by applying an inverse Laplace transform giving
\small \begin{eqnarray}  \label{eq:APP_joint_PDF_GSC_4}
\!\!\!\!\!\!\!\!\!\!\!\! && \!\!\!\!\!\!\!\! p_Z \left( {z_1 ,z_2 ,z_3 } \right) = \mathcal{L}_{S_1 ,S_2 ,S_3 }^{ - 1} \left\{ {MGF_Z \left( { - S_1 , - S_2 , - S_3 } \right)} \right\} \nonumber
\\ 
\!\!\!\!\!\!\!\!\!\!\!\! &=& \!\!\!\! \frac{F}{{\left( {K_s  - 2} \right)!}}\int\limits_0^\infty  {d\gamma _{K_s :K} p\left( {\gamma _{K_s :K} } \right)\mathcal{L}_{S_3 }^{ - 1} \left\{ {\exp \left( { - S_3 \gamma _{K_s :K} } \right)} \right\}\left[ {c\left( {\gamma _{K_s :K} } \right)} \right]^{\left( {K - K_s } \right)} }  \nonumber
\\ 
\!\!\!\!\!\!\!\!\!\!\!\! && \!\!\!\!  \times \!\!\!\! \int\limits_{\gamma _{K_s :K} }^\infty  {d\gamma _{K_s  - 1:K} p\left( {\gamma _{K_s  - 1:K} } \right)\mathcal{L}_{S_2 }^{ - 1} \left\{ {\exp \left( { - S_2 \gamma _{K_s  - 1:K} } \right)} \right\}\mathcal{L}_{S_1 }^{ - 1} \left\{ {\left[ {e\left( {\gamma _{K_s  - 1:K} , - S_1 } \right)} \right]^{\left( {K_s  - 2} \right)} } \right\}}  \nonumber
\\ 
\!\!\!\!\!\!\!\!\!\!\!\! &=& \!\!\!\! \frac{F}{{\left( {K_s  - 2} \right)!}}p\left( {z_2 } \right)p\left( {z_3 } \right)\left[ {c\left( {z_3 } \right)} \right]^{\left( {K - K_s } \right)} U\left( {z_2  - z_3 } \right)\mathcal{L}_{S_1 }^{ - 1} \left\{ {\left[ {e\left( {z_2 , - S_1 } \right)} \right]^{\left( {K_s  - 2} \right)} } \right\}.
\end{eqnarray} \normalsize

With these 3-dimensional joint PDF, letting $X=Z_2$ and $Y=Z_1+Z_3$ we can obtain the 2-dimensional joint PDF of $Z^{'}=[X,Y]$ by integrating over $z_3$ yielding
\begin{equation} \label{eq:AP_final_4}
\small p_{Z^{'}} \left( {x,y} \right) = \int_0^x {p_Z \left( {y - z_3 ,x,z_3 } \right)dz_3 } ,
\end{equation}
or equivalently we can obtain the 2-dimensional joint PDF of $Z^{'}=[X,Y]$ by integrating over $z_1$ yielding
\begin{equation} \label{eq:AP_final_5}
\small p_{Z^{'}} \left( {x,y} \right) = \int_{\left( {K_s  - 2} \right)x}^y {p_Z \left( {z_1 ,x,y - z_1 } \right)dz_1 } .
\end{equation}

Finally, we now consider the case iv) for which  $m=K_s$. Let $Z_1  = \gamma _{K_s :K}$ and $Z_2  = \sum\limits_{n = 1}^{K_s  - 1} {\gamma _{n:K} }$. For this case, the 2-dimensional MGF of $Z=\left[Z_1,Z_2\right]$ is given by the expectation
\small \begin{eqnarray} \label{eq:APP_joint_MGF_GSC_5_integralform}
\!\!\!\!\!\!\!\!\!\!\!\! && \!\!\!\!\!\!\!\! MGF_Z \left( {\lambda _1 ,\lambda _2 } \right) = E\left\{ {\exp \left( {\lambda _1 Z_1  + \lambda _2 Z_2 } \right)} \right\}  \nonumber
\\ 
\!\!\!\!\!\!\!\!\!\!\!\! &=& \!\!\!\! F\int\limits_0^\infty  {d\gamma _{1:K} p\left( {\gamma _{1:K} } \right)\exp \left( {\lambda _2 \gamma _{1:K} } \right) \cdots \!\!\!\! \int\limits_0^{\gamma _{K_s  - 2:K} } {d\gamma _{K_s  - 1:K} p\left( {\gamma _{K_s  - 1:K} } \right)\exp \left( {\lambda _2 \gamma _{K_s  - 1:K} } \right)} }  \nonumber
\\ 
\!\!\!\!\!\!\!\!\!\!\!\! && \!\!\!\!  \times \!\!\!\! \int\limits_0^{\gamma _{K_s  - 1:K} } {d\gamma _{K_s :K} p\left( {\gamma _{K_s :K} } \right)\exp \left( {\lambda _1 \gamma _{K_s :K} } \right)\left[ {c\left( {\gamma _{K_s :K} } \right)} \right]^{\left( {K - K_s } \right)} } . 
\end{eqnarray} \normalsize
With the help of (\ref{eq:Integral_solution}) and (\ref{eq:EDF_MGF_multiple}), we can easily obtain the 2-dimensional MGF of $Z_1  = \gamma _{K_s :K}$ and $Z_2  = \sum\limits_{n = 1}^{K_s  - 1} {\gamma _{n:K} }$ as
\begin{equation} \label{eq:APP_joint_MGF_GSC_5}\small 
 MGF_Z \left( {\lambda _1 ,\lambda _2 } \right) \!\! = \!\! \frac{F}{{\left( {K_s  - 1} \right)!}}\int\limits_0^\infty  {d\gamma _{K_s :K} p\left( {\gamma _{K_s :K} } \right)\exp \left( {\lambda _1 \gamma _{K_s :K} } \right)\left[ {c\left( {\gamma _{K_s :K} } \right)} \right]^{\left( {K - K_s } \right)} }  \left[ {e\left( {\gamma _{K_s :K} ,\lambda _2 } \right)} \right]^{\left( {K_s  - 1} \right)}.
\end{equation} \normalsize

Having an MGF expression as given in (\ref{eq:APP_joint_MGF_GSC_5}), we are now in the position to derive the 3-dimensional joint PDF of $Z_1  = \gamma _{K_s :K}$ and $Z_2  = \sum\limits_{n = 1}^{K_s  - 1} {\gamma _{n:K} }$. Letting $\lambda _1  =  - S_1$ and $\lambda _2  =  - S_2$  we can derive the 2-dimensional joint PDF of $Z_1  = \gamma _{K_s :K}$ and $Z_2  = \sum\limits_{n = 1}^{K_s  - 1} {\gamma _{n:K} }$ by applying an inverse Laplace transform yielding
\small \begin{eqnarray}  \label{eq:APP_joint_PDF_GSC_5}
\!\!\!\!\!\!\!\!\!\!\!\! && \!\!\!\! p_Z \left( {z_1 ,z_2 } \right) = \mathcal{L}_{S_1 ,S_2 }^{ - 1} \left\{ {MGF_Z \left( { - S_1 , - S_2 } \right)} \right\} \nonumber
\\ 
\!\!\!\!\!\!\!\!\!\!\!\! &=& \!\!\!\! \frac{F}{{\left( {K_s  - 1} \right)!}}\int\limits_0^\infty  {d\gamma _{K_s :K} p\left( {\gamma _{K_s :K} } \right)\mathcal{L}_{S_1 }^{ - 1} \left\{ {\exp \left( { - S_1 \gamma _{K_s :K} } \right)} \right\}\left[ {c\left( {\gamma _{K_s :K} } \right)} \right]^{\left( {K - K_s } \right)} }  \nonumber
\\ 
\!\!\!\!\!\!\!\!\!\!\!\! && \!\!\!\!  \times \mathcal{L}_{S_2 }^{ - 1} \left\{ {\left[ {e\left( {\gamma _{K_s :K} , - S_2 } \right)} \right]^{\left( {K_s  - 1} \right)} } \right\} \nonumber
\\ 
\!\!\!\!\!\!\!\!\!\!\!\! &=& \!\!\!\!\frac{F}{{\left( {K_s  - 1} \right)!}}p\left( {z_1 } \right)\left[ {c\left( {z_1 } \right)} \right]^{\left( {K - K_s } \right)} \mathcal{L}_{S_2 }^{ - 1} \left\{ {\left[ {e\left( {z_1 , - S_2 } \right)} \right]^{\left( {K_s  - 1} \right)} } \right\}. 
\end{eqnarray} \normalsize

\clearpage
\bibliographystyle{ieeetran}

\clearpage
\begin{figure}
\centering
\subfigure[Example of original $M$-dimensional groups\label{Example_1_1}]{\includegraphics[width=5.5in,trim=0.5cm 1.5cm 0.5cm 0cm]{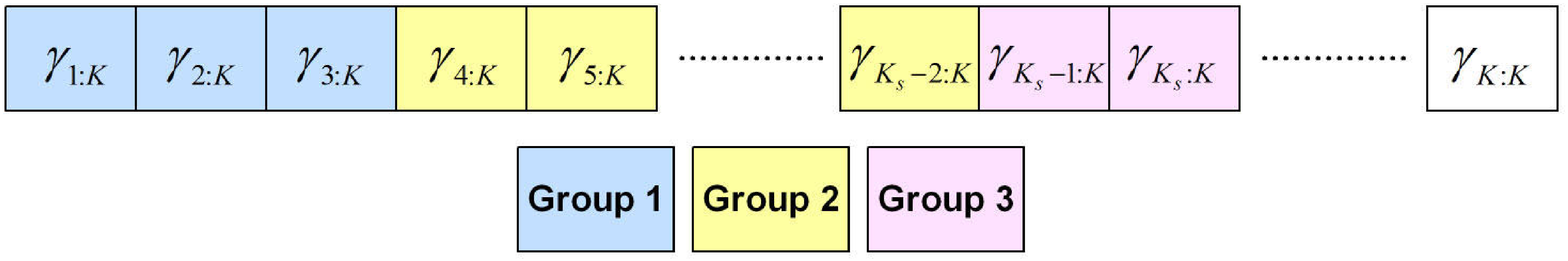}}\\
\subfigure[Example of substituted $(M+1)$-dimensional groups\label{Example_1_2}]{\includegraphics[width=5.5in,trim=0.5cm 1.5cm 0.5cm 0cm]{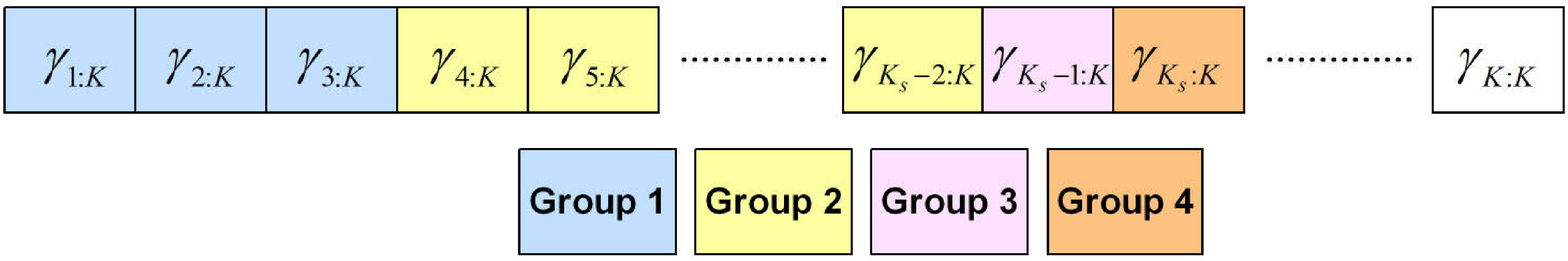}}
\caption{Examples for 3-dimensional joint PDF with non-split groups.}
\label{Example_1}
\end{figure}

\begin{figure}
\centering
\subfigure[Example of original $M$-dimensional groups\label{Example_a}]{\includegraphics[width=5.5in,trim=0.5cm 1.5cm 0.5cm 0cm]{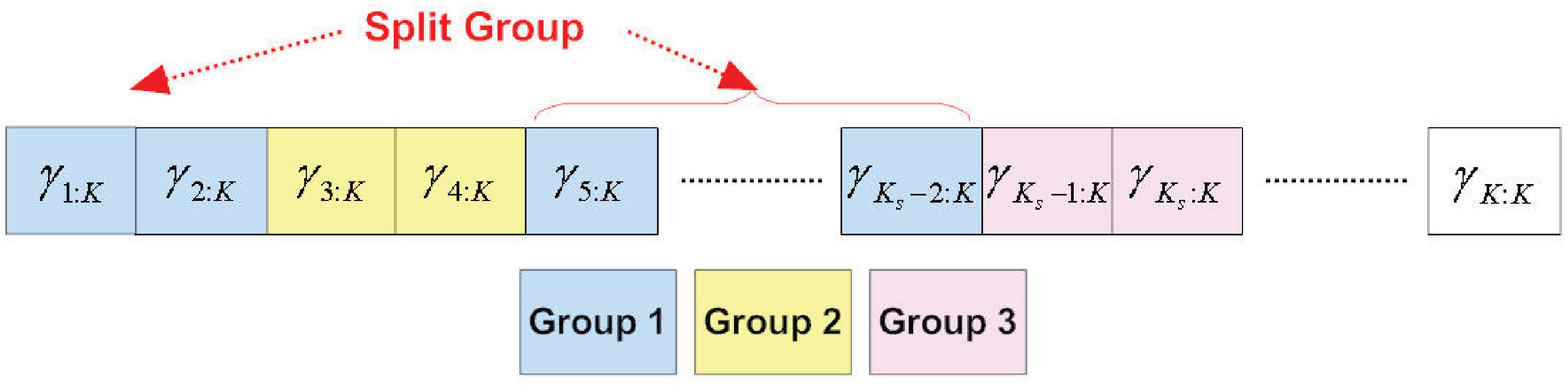}}\\
\subfigure[Example of substituted split groups\label{Example_b}]{\includegraphics[width=5.5in,trim=0.5cm 1.5cm 0.5cm 0cm]{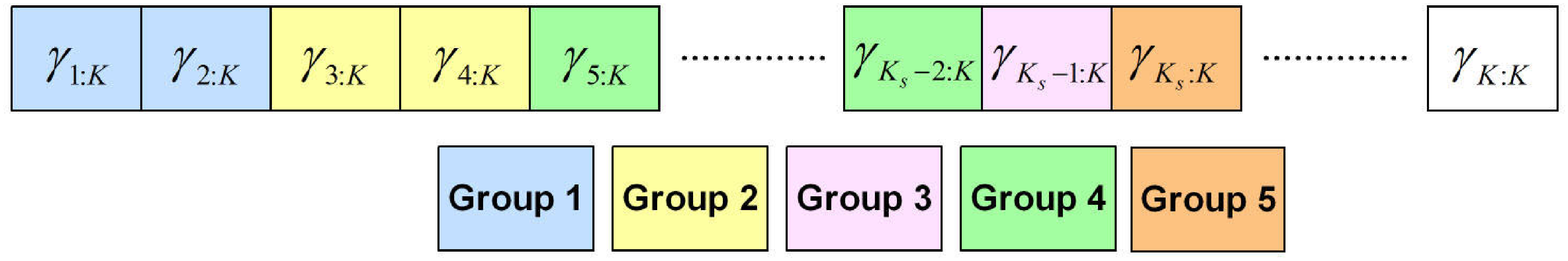}}
\caption{Examples for 3-dimensional joint PDF with split groups.}
\label{Example_2}
\end{figure}

\begin{figure}[h!]
\centering
\includegraphics[width=5.5in]{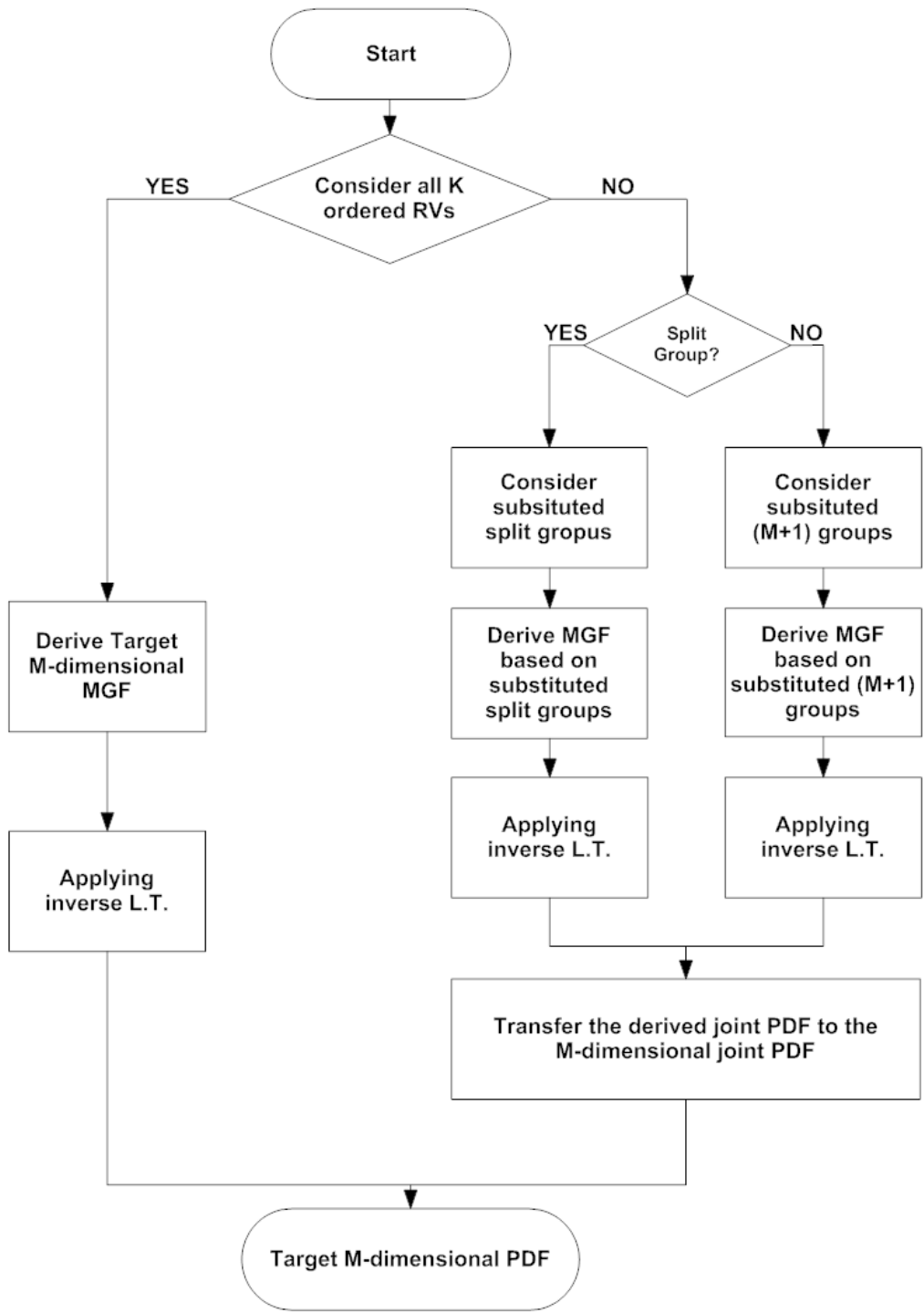}
\caption{Flow chart to obtain the desired $M$-dimensional joint PDF.}
\label{flowchart}
\end{figure}

\begin{figure}
\centering
\subfigure[For the case of $m=1$\label{Example_3_1}]{\includegraphics[width=5.5in,trim=0.5cm 1.5cm 0.5cm 0cm]{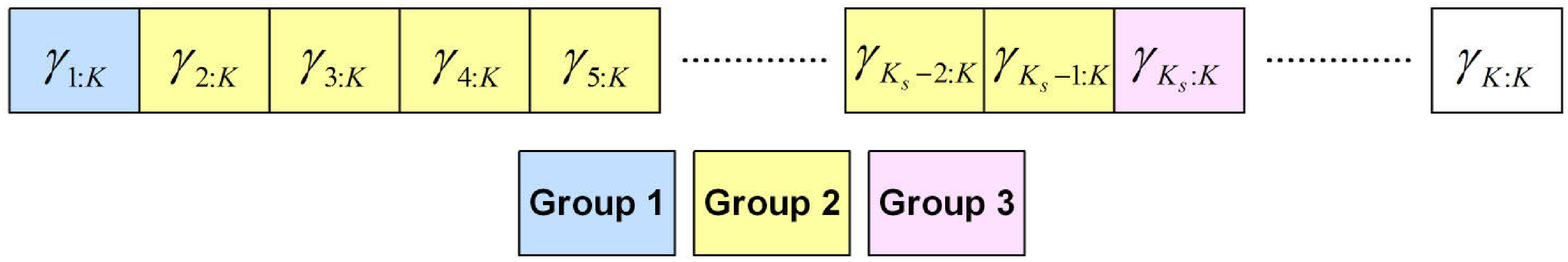}}\\
\subfigure[For the case of $1<m<K_s-1$\label{Example_3_2}]{\includegraphics[width=5.5in,trim=0.5cm 1.5cm 0.5cm 0cm]{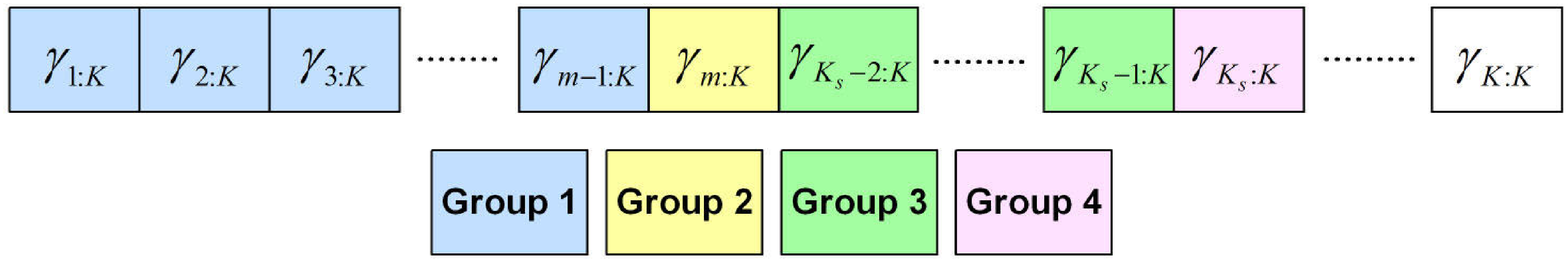}}\\
\subfigure[For the case of $m=K_s-1$\label{Example_3_3}]{\includegraphics[width=5.5in,trim=0.5cm 1.5cm 0.5cm 0cm]{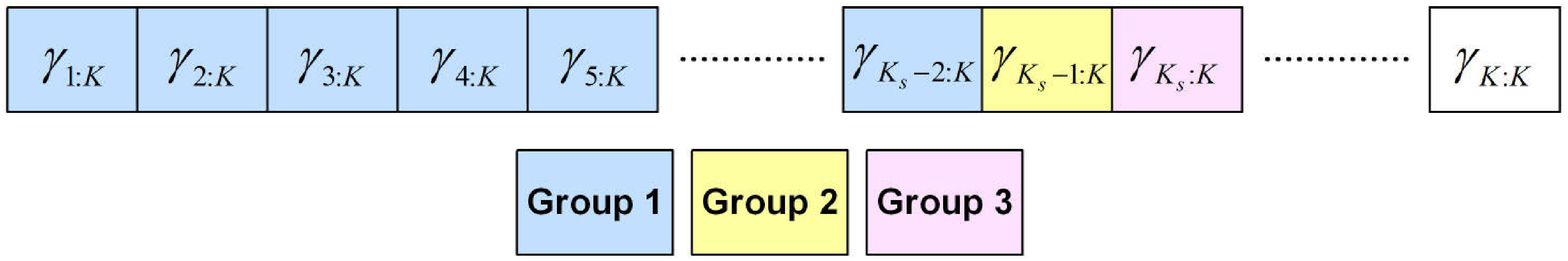}}\\
\subfigure[For the case of $m=K_s$\label{Example_3_4}]{\includegraphics[width=5.5in,trim=0.5cm 1.5cm 0.5cm 0cm]{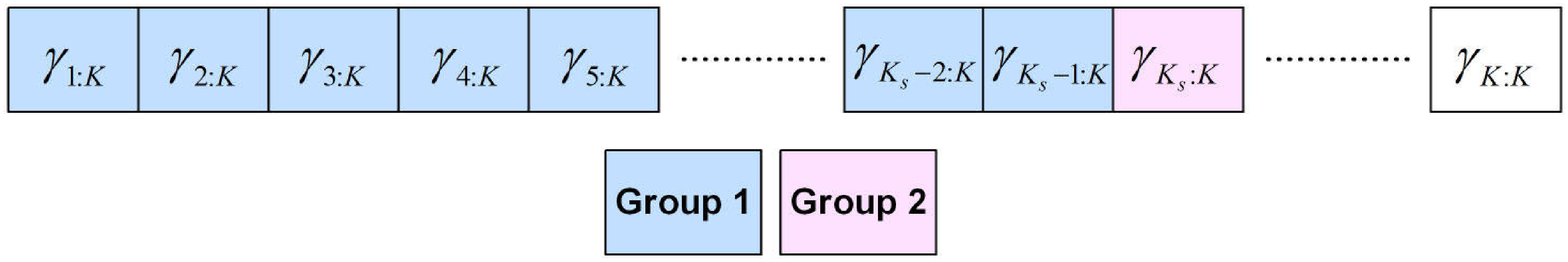}}
\caption{Joint MGF of $\gamma _{m:K}$ and $\sum\limits_{\scriptstyle n = 1 \hfill \atop \scriptstyle n \ne m \hfill}^{K_s } {\gamma _{n:K} }$.}
\label{Example_3}
\end{figure}

\clearpage
\begin{biographynophoto}{Sung Sik Nam} (S'05, M'09) was born in Korea. He received the B.S. and M.S. degrees in Electronic Engineering from Hanyang University, Korea, in 1998 and 2000, respectively. Also He received the M.S. degree in Electrical Engineering from University of Southern California (USC), Los Angeles, CA, USA, in 2003, and the Ph.D. degree at Texas A\&M University (TAMU), College Station, TX, USA, in 2009. From 1998 to 1999, he worked as a researcher at The Electronics \& Telecommunication Research Institute (ETRI), Daejeon, Korea. From 2003 through 2004, he worked as a manager at the Korea Telecom Corporation (KT), New Business Planning Office, New Business Planning Team, Korea. Since June 2009, he has been with the Department of Electronic Engineering at Hanyang University, Korea. His research interests include the design and performance analysis of wireless communications, diversity techniques, power control, multiuser scheduling.
\end{biographynophoto}
\begin{biographynophoto}{Mohamed-Slim Alouini} (S'94, M'98, SM'03, F'09) was
born in Tunis, Tunisia. He received the Ph.D. degree in electrical engineering
from the California Institute of Technology (Caltech), Pasadena,
CA, USA, in 1998. He was with the department of Electrical and Computer Engineering of the University of Minnesota,Minneapolis, MN, USA, then with the Electrical and Computer Engineering Program at the Texas A\&M University at Qatar, Education City, Doha, Qatar. Since June 2009, he has been a Professor of Electrical Engineering in the Division of
Physical Sciences and Engineering at KAUST, Saudi Arabia., where his current research interests include the design and performance analysis of wireless communication systems.
\end{biographynophoto}
\begin{biographynophoto}{Hong-Chuan Yang} (S'00, M'03, SM'07) was born in Liaoning, China.
He received the B.E. degree from the
Changchun Institute of Posts \& Telecomm. (now part of Jilin University),
Changchun, China, in 1995. During his graduate study at
the University of Minnesota, USA, Dr. Yang received an
M.Sc. degree in Applied and Computational Mathematics in 2000,
an M.Sc. degree in Electrical Engineering in 2001, and a Ph.D. degree
in Electrical Engineering in 2003.
Dr. Yang has been with the Department of Electrical and Computer
Engineering Department of the University of Victoria, Victoria, B.C., Canada since September 2003.
He became an associate professor in July 2007.
From 1995 to 1998, Dr. Yang was a Research Associate at the
Science and Technology Information Center (STIC) of
Ministry of Posts \& Telecomm. (MPT), Beijing, China.
His research work focuses on different aspects of
wireless communications, with special emphasis on channel modeling,
diversity techniques, system performance evaluation, cross-layer design
and energy efficient communications.
Dr. Yang is a recipient of the Doctoral Dissertation Fellowship (DDF)
Award from the graduate school of the University of Minnesota
for the 2002-2003 academic year. He is serving as an editor for {\it IEEE Transactions on Wireless Communications} and {\it Journal on Communications and Networking}. Dr. Yang is a registered professional engineer in B. C., Canada.
\end{biographynophoto}
\end{document}